\documentclass[sort&compress]{elsarticle}
\usepackage{amsmath,mathtools,amssymb}
\usepackage{graphicx}   
\usepackage{epstopdf}
\usepackage{subfigure}
\usepackage{color}
\usepackage{verbatim} 
\usepackage{booktabs}
\usepackage{bm}   
\usepackage{dutchcal}
\usepackage[left=3cm,right=3cm,top=2.5cm,bottom=2.5cm]{geometry}  
\usepackage{lineno,hyperref}
\usepackage{tablists} 
\usepackage{makecell} 
\usepackage{tikz}
\usetikzlibrary{arrows.meta, positioning}
\usepackage{hyperref}
\usepackage{pstricks}
\usepackage[titletoc]{appendix}
\usepackage{pstricks-add}
\usepackage{float}
\usepackage{mathrsfs} 
\usepackage{amsthm}
\usepackage{threeparttable}
\hypersetup{
	colorlinks=ture,
	linkcolor=blue,
	filecolor=gray,
	urlcolor=blue,
	citecolor=blue}
\numberwithin{equation}{section}
\newtheorem{definition}{Definition}[section]

\newtheorem{theorem}{Theorem}[section]
\newtheorem{lemma}{Lemma}[section]	
\newtheorem{proposition}{Proposition}[section]
\allowdisplaybreaks[4]

\begin{document}
	\begin{frontmatter}
		\title{Generalized  forms of types $N=1, 2$ and higher gauge theory} 
	
		\author[a]{Danhua Song \corref{cor1}}
		\ead{songdh@pku.edu.cn}
		\author[b]{Mengyao Wu}
		\ead{2210502108@cnu.edu.cn}
		\cortext[cor1]{Corresponding author.}
		\address[a]{Beijing International Center for Mathematical Research, Peking University, Beijing, China}
		\address[b]{School of Mathematical Sciences, Capital Normal University, Beijing, China}
		\date{}
		
	\begin{abstract}
		In this paper, we give a compact formulation of strict higher gauge theory based on generalized (differential) forms that package fields of multiple form degrees into a single variable.  We define generalized forms valued in higher algebras and higher groups and derive the corresponding Maurer–Cartan structures. This leads to uniform, gauge-theory-like expressions for higher connections, curvatures, Bianchi identities, and gauge transformations. We further construct action principles for higher Chern–Simons and higher Yang–Mills theories within the same formalism and compute the associated topological densities in the corresponding dimensions.
	\end{abstract}
		\begin{keyword}
			Generalized Differential Calculus, Higher Gauge Theory, Maurer--Cartan Forms, Higher Gauge Transformations
		\end{keyword}	
	\end{frontmatter}

		\tableofcontents

	\section{Introduction}
	Higher gauge theory \cite{Baez.2010,FH,Baez2005HigherGT} extends ordinary gauge theory by encoding gauge potentials and their curvatures in differential forms of degree $>1$. It provides a natural framework for describing higher-form gauge fields carried by extended objects and for organizing higher-form symmetries in quantum field theory and in string and M-theory (see  Ref.~\cite{Borsten.2024} for a recent review). In such settings the kinematics is described by several differential-form potentials of increasing degree, together with their higher curvatures and multi-degree gauge transformations. Although the general formalism is by now well established, especially for strict 2- and 3-gauge theory \cite{JCBADL,JCBASC,JBUS,doi:10.1063/1.4870640,JFM,TRMV}, the standard componentwise presentation tends to obscure the underlying uniform pattern and can make explicit computations cumbersome.
	
	The mathematical formulation of these higher theories relies on higher algebraic and geometric structures. Algebraically, strict 2-gauge theory is governed by a Lie crossed module, equivalently a strict Lie 2-group \cite{JCBADL,JCBASC}. It is specified by Lie groups $H$ and $G$, a homomorphism $\bar{\alpha}:H\to G$, and a smooth action $\bar{\triangleright}:G\times H\to H$ by automorphisms.
	Strict 3-gauge theory is similarly described by a Lie 2-crossed module \cite{JFM,TRMV}, specified by Lie groups $L$, $H$, and $G$, a complex $L \xrightarrow{\bar{\beta}} H \xrightarrow{\bar{\alpha}} G$, compatible smooth $G$-actions on $L$ and $H$, and a Peiffer lifting $\{\cdot,\cdot\}:H\times H\to L$. Geometrically, principal bundles are generalized to principal $n$-bundles (or higher principal bundles) \cite{PLB,Wockel,Stevenson,BJ,Nikolaus11,Bunk}, yielding the differential-geometric framework for higher connections, curvatures, and gauge transformations. For orientation, the corresponding local field content and gauge parameters for $n=2,3$ are summarized in Table~\ref{table1}.
	\begin{table*}[h]
		\begin{threeparttable}
			\caption{The gauge structures of (higher) gauge theories} \label{table1}
			\setlength{\tabcolsep}{4.8mm}{
				\begin{tabular}{cccc}
					\specialrule{1pt}{5pt}{1pt} \specialrule{0.5pt}{1pt}{3pt}
					\textbf{Structure} & \textbf{Gauge Theory}&\textbf{2-Gauge Theory}& \textbf{3-Gauge Theory}  \\[1mm] 
					\specialrule{0.5pt}{2pt}{6pt} 
					gauge group &   \begin{tabular}[c]{@{}l@{}}  Lie group\\ \ \ \ \ \ $G$   \end{tabular}  
					& \begin{tabular}[c]{@{}l@{}}  Lie 2-group    \\
						$(H, G; \bar{\alpha}, \bar{\triangleright})$ \end{tabular}   
					&  \begin{tabular}[c]{@{}l@{}}  \ \ \ \ \ \ \ Lie 3-group  \\
						$(L, H, G;\bar{\beta}, \bar{\alpha}, \bar{\triangleright}, \left\{\cdot,\cdot\right\})$  \end{tabular}   \\[1.5mm]
					\specialrule{0.5pt}{1pt}{6pt} 
					connection & 
					$A$: $\mathfrak{g}$-valued 1-form  &\begin{tabular}[c]{@{}l@{}} $A$: $\mathfrak{g}$-valued 1-form  \\
						$B$: $\mathfrak{h}$-valued 2-form  \end{tabular}   &  \begin{tabular}[c]{@{}l@{}l@{}} $A$: $\mathfrak{g}$-valued 1-form  \\
						$B$: $\mathfrak{h}$-valued 2-form  \\
						$C$: $\mathfrak{l}$-valued 3-form\end{tabular} 
					\\[1.5mm]
					\specialrule{0.5pt}{1pt}{6pt} 
					curvature & 	$F$: $\mathfrak{g}$-valued 2-form  & \begin{tabular}[c]{@{}l@{}} 
						$\Omega_1$: $\mathfrak{g}$-valued 2-form  \\
						$\Omega_2$: $\mathfrak{h}$-valued 3-form  \end{tabular} & 
					\begin{tabular}[c]{@{}l@{}l@{}} $\Omega_1$: $\mathfrak{g}$-valued 2-form  \\
						$\Omega_2$: $\mathfrak{h}$-valued 3-form  \\
						$\Omega_3$: $\mathfrak{l}$-valued 4-form\end{tabular} 
					\\[1.5mm]
					\specialrule{0.5pt}{1pt}{6pt} 
					\begin{tabular}[c]{@{}l@{}}
						\ \ \ \ \ gauge \\
						transformation \end{tabular} & 
					$g$: $G$-valued 0-form  &\begin{tabular}[c]{@{}l@{}} $g$: $G$-valued 0-form  \\
						$\phi$: $\mathfrak{h}$-valued 1-form  \end{tabular}   &  \begin{tabular}[c]{@{}l@{}l@{}} $g$: $G$-valued 0-form  \\
						$\phi$: $\mathfrak{h}$-valued 1-form  \\
						$\psi$: $\mathfrak{l}$-valued 2-form\end{tabular} 
					\\[2.5mm]  
					\specialrule{0.5pt}{2pt}{1pt} \specialrule{1pt}{1pt}{0pt} 
			\end{tabular} }
			\begin{tablenotes}
				\small\item Note: The symbols $\mathfrak{g}$, $\mathfrak{h}$, and $\mathfrak{l}$ denote the Lie algebras corresponding to the Lie groups $G$, $H$, and $L$, respectively.
			\end{tablenotes}
		\end{threeparttable}
	\end{table*}
	The last row lists the gauge transformation parameters. The table highlights a basic feature of higher gauge theory: its local kinematics is encoded in several differential forms of increasing degree, valued in the corresponding higher algebra. While higher categorical languages \cite{Leinster,Simpson} provide one conceptual approach to these structures, they do not by themselves furnish a uniform calculus well suited to concrete computations.

	Higher gauge theory also provides a natural arena for higher extensions of familiar gauge models. Yang–Mills-type dynamics admits higher-form generalizations \cite{Pfeiffer,Henneaux,Gastel,JCB02,Song1,SDH-BFYM}, and Chern–Simons theory admits higher analogues in higher dimensions \cite{Z-2021,Zucchini-2014-1,Zucchini-2016,Zucchini-2013,SDH4}. Although both higher Yang–Mills (HYM) and higher Chern–Simons (HCS) theories can be formulated within strict 2- and 3-gauge theory, existing constructions often employ different methods and notational conventions. This motivates the search for a unified framework in which HYM and HCS can be developed in parallel.
	
	A convenient tool for organizing multi-degree data is \emph{generalized differential calculus} (GDC) \cite{2007R3,2001NR,2002NR,1998GAJ}, an extension of Cartan's calculus of differential forms \cite{Cartan0,Cartan2,Cartan5}. In GDC, ordinary forms of different degrees are packaged into a single generalized form. The resulting algebraic and differential rules closely parallel those of ordinary forms, while allowing important new features: the theory depends on a discrete \emph{type} $N$, generalized de Rham cohomology depends on $N$, and forms of negative degree are permitted. The notion of a $(-1)$-form was introduced by Sparling in the context of twistor theory \cite{1998GAJ,1998ZPGAJS,GAJ}. Nurowski and Robinson \cite{2001NR,2002NR} developed the corresponding calculus, including generalized Cartan structure equations, Hodge operators, codifferentials, and Laplacians.  Robinson \cite{2003R1} later extended the formalism to arbitrary type~$N$, where $N$ counts the number of independent $(-1)$-form fields. Concretely, a generalized $p$-form of type $N$ can be represented as an ordered $(N+1)$-tuple of ordinary forms of degrees $p,p+1,\dots,p+N$. Besides this tuple description, base \cite{2007R3} and matrix \cite{2003R1} representations are available and are useful in different computational settings.
	
	Generalized forms have been applied to a range of geometric and physical contexts. In particular, connections and gauge theories admit generalized formulations \cite{2003R2,2002GLTZ}. Moreover, several field theories, such as BF theory, Yang–Mills theory, and (super)gravity, have been recast as generalized topological field theories, with generalized Chern--Pontrjagin and Chern–Simons forms as Lagrangians \cite{2007R3,2004LTG,2009R,2013R}. Generalized forms also admit an interpretation in terms of forms on path space \cite{2008SCALANS}. These examples illustrate that generalized forms provide a unified tool for computations in mathematical physics.
	
	The structure of generalized forms is aligned with the multi-degree structure of higher gauge theory. This leads to two natural questions:
	\begin{description}
		\item[(i)] Can higher gauge kinematics (including gauge transformations) be uniformly described in terms of generalized forms, in direct analogy with generalized gauge theory?
		\item[(ii)] Can one construct a simple unified action principle for HCS and HYM theories within this same formalism?
	\end{description}
	Previous work \cite{SDH4} laid foundations by encoding higher connections $(A,B)$ and $(A,B,C)$ into generalized 1-forms and their curvatures into generalized 2-forms, thereby revealing a close structural analogy with ordinary gauge theory. However, two key elements remained missing: a unified treatment of higher gauge transformations and a systematic derivation of HCS and HYM actions within one framework.
	
	In this paper we address both questions. We develop the calculus of generalized forms valued in higher algebras and higher groups, derive the associated Maurer–Cartan structures, and cast higher connections, curvatures, Bianchi identities, and gauge transformations into compact expressions paralleling ordinary gauge theory. Building on this uniform formalism, we construct action principles for HCS and HYM theories and compute the corresponding topological densities in the relevant dimensions.
	
	The article is organized as follows. Section~\ref{sec2} recalls the formalism of generalized $p$-forms of type~$N$, including their exterior products and derivatives, with emphasis on the base representation (Subsection~\ref{sec2.1}). Building on this, Subsection~\ref{sec2.2} introduces generalized gauge theory, focusing on generalized connections and their gauge transformations.
	Section~\ref{sec3} adapts the calculus to higher gauge fields. For  Lie 2- and Lie 3-algebras, we develop the calculus of higher algebra-valued generalized forms of types $N=1$ and $N=2$, defining the required bilinear brackets, exterior derivatives, pairings and inner products. Similarly, we construct higher group-valued generalized 0-forms, define the corresponding higher Maurer–Cartan forms, and establish the higher Maurer–Cartan equations.
	In Section~\ref{sec4}, we reformulate standard higher gauge structures, including connections, curvatures, Bianchi identities and gauge transformations, by using the higher generalized forms developed in Section~\ref{sec3}.
	Section~\ref{sec5.1} uses higher generalized forms of types $N=1,2$ to construct the 4D 2-Chern–Simons (2CS)  and  5D 3-Chern–Simons (3CS) forms, together with their topological invariants (the 5D 2-Chern and 6D 3-Chern forms). Section~\ref{sec5.2} provides a generalized construction of the action functionals for 2-form and 3-form Yang–Mills (2YM and 3YM) theories.
	Finally, Section~\ref{sec-6} summarizes the main findings of this work and outlines several promising avenues for future investigation.

	\section{Generalized  forms and gauge theory}\label{sec2}
	\subsection{Generalized  forms of type $N$}\label{sec2.1}
	We work on a smooth, oriented, real $n$-dimensional manifold $M$. Throughout this paper we follow the notational conventions of Robinson \cite{2007R3,2009R}, with minor modifications. Unless stated otherwise, all geometric objects are smooth. In particular:
	\begin{itemize}
		\item \textbf{Ordinary differential forms} are denoted by Roman letters with a superscript indicating the degree, e.g.\ $\overset{p}{a}$.
		
		\item The \textbf{wedge product} is written as $\overset{p}{a}\wedge \overset{q}{b}$, or simply $\overset{p}{a}\,\overset{q}{b}$ when no confusion can arise.
		
		\item An ordinary form $\overset{p}{a}$ is understood to be \textbf{zero} whenever $p<0$ or $p>\dim(M)$.
		
		\item \textbf{Generalized forms} are denoted by calligraphic capital letters with a degree superscript, e.g.\ $\overset{p}{\mathcal{A}}$. The type $N$ is indicated by a subscript, as in $\overset{p}{\mathcal{A}}_{(N)}$; both the superscript and the subscript are omitted when unambiguous.
		
		\item The space of ordinary $p$-forms on $M$ is denoted by $\Lambda^{p}(M)$, and the space of generalized $p$-forms of type $N$ by $\Lambda^{p}_{(N)}(M)$.
	\end{itemize}
	We restrict attention to real geometry, and all forms and coefficients are real-valued.

	\subsubsection{Basic properties  and formalism}
	This subsection introduces the fundamental formalism for generalized forms, comprising their representation in a basis, their exterior (wedge) product, and a symmetric inner product. These ingredients underpin the generalized constructions in the higher gauge theory developed below.
	
	When $N=0$, a generalized $p$-form is just an ordinary differential $p$-form. For $N\ge 1$, a generalized $p$-form admits a unique expansion in an augmented basis consisting of an arbitrary basis of ordinary differential forms on $M$ together with $N$ linearly independent $(-1)$-forms $\{\xi^i\}_{i=1}^N$. These $(-1)$-forms are assigned the same algebraic properties as ordinary forms, except that their degree is $-1$. In particular, they satisfy:
	\begin{enumerate}
		\item the distributive and associative laws of an exterior algebra;
		\item graded commutativity,
		\[
		\xi^i\xi^j=-\,\xi^j\xi^i,
		\qquad
		\overset{p}{a}\,\xi^i = (-1)^p\,\xi^i\,\overset{p}{a};
		\]
		\item linear independence, in the sense that $\xi^1\xi^2\cdots\xi^N\neq 0$.
	\end{enumerate}
	Here juxtaposition (e.g.\ $\xi^i\xi^j$ or $\overset{p}{a}\,\xi^i$) denotes the exterior product. Note that $\xi^i$ is not a differential form in the usual sense; for instance, $\overset{1}{a}\,\xi$ is not a function on $M$.

	\begin{definition}[Generalized  forms of type $N$]
		A generalized $p$-form of type $N \geq 1$ on $M$, denoted $\overset{p}{\mathcal{A}}_{(N)} \in \Lambda^p_{(N)}(M)$,  is a geometric object admitting the unique expansion
		\begin{align}\label{define-gfN}
			\overset{p}{\mathcal{A}}_{(N)} =& \overset{p}{a} +\overset{p+1}{a}_{i_1} \xi^{i_1}+ \frac{1}{2!}\overset{p+2}{a}_{i_1 i_2} \xi^{i_1}  \xi^{i_2} + \cdots + \frac{1}{j!}\overset{p+j}{a}_{i_1\dots i_j} \xi^{i_1}\nonumber\\
			& \dots \xi^{i_j} + \cdots + \frac{1}{N!}\overset{p+N}{a}_{i_1\dots i_N} \xi^{i_1} \dots \xi^{i_N},
		\end{align}
		where $\overset{p}{a}$ and $\overset{p+j}{a}_{i_1\dots i_j}$ are ordinary $p$- and $(p+j)$-forms, respectively. The indices $i_1,\dots,i_j$, $\dots$, $i_N$ take values in ${1,\dots,N}$ and are summed over, and the coefficient forms are antisymmetric in their indices, $\overset{p+j}{a}_{i_1\dots i_j}=\overset{p+j}{a}_{[i_1\dots i_j]}$.
	\end{definition}
	
	The expansion~\eqref{define-gfN} shows that the degree $p$ of a  generalized form of type $N$ is constrained by
	\[
	-N \le p \le n,
	\]
	since the component $\overset{p+j}{a}_{i_1\cdots i_j}$ is an ordinary $(p+j)$-form on an $n$-dimensional manifold.  Relative to a fixed choice of basis $\{\xi^i\}$, a generalized form is therefore specified by an ordered $(N+1)$-tuple of ordinary forms,
	\[
	(\overset{p}{a},\,\overset{p+1}{a},\,\dots,\,\overset{p+N}{a}).
	\]
	
	A key feature of this description is its recursive structure: any generalized $p$-form of type $N$ can be written equivalently as a pair of generalized forms of type $(N-1)$,
	\begin{equation*}
		\overset{p}{\mathcal{A}}_{(N)}=\overset{p}{\mathcal{A}}_{(N-1)}+	\overset{p+1}{\mathcal{A}}_{(N-1)}\xi^N=\big( \overset{p}{\mathcal{A}}_{(N-1)},	\overset{p+1}{\mathcal{A}}_{(N-1)}\big).
	\end{equation*}
	Iterating this relation reconstructs the full expansion~\eqref{define-gfN} from the case $N=1$.
	By packaging a chain of ordinary forms of adjacent degrees into a single object, this recursive organization provides a natural language for higher gauge theory, in which the fundamental fields are intrinsically multi-form.
	
	This recursive viewpoint suggests inductive definitions of the algebraic operations.
	\begin{definition}[Exterior product]\label{def-wedg}
		Let $\overset{p}{\mathcal{A}}_{(N)}=	\overset{p}{\mathcal{A}}_{(N-1)}+	\overset{p+1}{\mathcal{A}}_{(N-1)}\xi^N$ and $\overset{q}{\mathcal{B}}_{(N)}=	\overset{q}{\mathcal{B}}_{(N-1)}+	\overset{q+1}{\mathcal{B}}_{(N-1)}\xi^N$. Their exterior product is defined to be the generalized $(p+q)$-form of type $N$ given by
		\begin{align*}
			\overset{p}{\mathcal{A}}_{(N)}\underline{\wedge} \overset{q}{\mathcal{B}}_{(N)}=\overset{p}{\mathcal{A}}_{(N-1)}\underline{\wedge}\overset{q}{\mathcal{B}}_{(N-1)}+ \big(\overset{p}{\mathcal{A}}_{(N-1)}\underline{\wedge} \overset{q+1}{\mathcal{B}}_{(N-1)}+(-1)^q\overset{p+1}{\mathcal{A}}_{(N-1)}\underline{\wedge}\overset{q}{\mathcal{B}}_{(N-1)} \big)\xi^N.
		\end{align*}
	\end{definition}
	
	From Definition \ref{def-wedg},  the product $\underline{\wedge}$ is graded-commutative, 
	\begin{equation*}
		\overset{p}{\mathcal{A}}_{(N)}\underline{\wedge} \overset{q}{\mathcal{B}}_{(N)}=(-1)^{pq}\overset{q}{\mathcal{B}}_{(N)}\underline{\wedge} \overset{p}{\mathcal{A}}_{(N)}.
	\end{equation*}
	Consequently,  $\Lambda_{(N)}(M)=\bigoplus_{p=-N}^{n}\Lambda^p_{(N)}(M)$ is a graded algebra with respect to $\underline{\wedge}$.

	A symmetric inner product is defined similarly by recursion \cite{2003R1}.
	For ordinary differential forms $\overset{p}{a}$ and $\overset{p}{b}$, we set
	\begin{equation}\label{inn0}
		(\!(\overset{p}{a}, \overset{p}{b})\!)=\int_M \overset{p}{a}\wedge * \overset{p}{b},
	\end{equation}
	where $*: \Lambda^p(M)\longrightarrow \Lambda^{n-p}(M)$ denotes the Hodge star operator.
	\begin{definition}[Symmetric inner product]\label{def-sip}
		For $\overset{p}{\mathcal{A}}_{(N)}=	\overset{p}{\mathcal{A}}_{(N-1)}+	\overset{p+1}{\mathcal{A}}_{(N-1)}\xi^N$ and $\overset{p}{\mathcal{B}}_{(N)}=	\overset{p}{\mathcal{B}}_{(N-1)}+	\overset{p+1}{\mathcal{B}}_{(N-1)}\xi^N$, the symmetric inner product is defined recursively by
		\begin{align}\label{sip}
			(\!(\overset{p}{\mathcal{A}}_{(N)}, \overset{p}{\mathcal{B}}_{(N)})\!)=	(\!(\overset{p}{\mathcal{A}}_{(N-1)}, \overset{p}{\mathcal{B}}_{(N-1)})\!)+	(\!(\overset{p+1}{\mathcal{A}}_{(N-1)}, \overset{p+1}{\mathcal{B}}_{(N-1)})\!),
		\end{align}
		with the initial condition $N=0$, which reduces to the ordinary inner product \eqref{inn0}.
	\end{definition}
	
	Detailed studies of generalized forms are often restricted to types $N \leq 2$. Here we likewise focus on the cases $N=1$ and $N=2$, which simplify the notation while revealing patterns that generalize to higher types.
	\paragraph{Type $N=1$.}Consider the generalized forms
	\begin{equation*}
		\overset{p}{\mathcal{A}}_{(1)}=\overset{p}{a}+ \overset{p+1}{a}\xi,\qquad \overset{q}{\mathcal{B}}_{(1)}=\overset{q}{b}+ \overset{q+1}{b}\xi.
	\end{equation*}
	The exterior product in Definition~\ref{def-wedg} becomes
	\begin{equation}\label{N=1:wedge}
		\overset{p}{\mathcal{A}}_{(1)}\underline{\wedge} \overset{q}{\mathcal{B}}_{(1)}=\overset{p}{a}\wedge\overset{q}{b} +\big(\overset{p}{a}\wedge \overset{q+1}{b} +(-1)^q \overset{p+1}{a} \wedge \overset{q}{b} \big)\xi.
	\end{equation} 
	For forms $\mathcal{A}_{(1)}$ and $\mathcal{B}_{(1)}$ of the same degree $p$, the symmetric inner product in Definition \ref{def-sip} reduces to
	\begin{equation}\label{inn1}
		(\!(\overset{p}{\mathcal{A}}_{(1)}, \overset{p}{\mathcal{B}}_{(1)})\!)=\int_M \big(\overset{p}{a}\wedge * \overset{p}{b}+ \overset{p+1}{a}\wedge * \overset{p+1}{b}\big).
	\end{equation}
	
	\paragraph{Type $N=2$.}
	Consider
	\begin{align*}
		\overset{p}{\mathcal{A}}_{(2)}&=\overset{p}{a}+ \overset{p+1}{a_1}\xi^1 +  \overset{p+1}{a_2}\xi^2+  \overset{p+2}{a}\xi^1\xi^2,\\
		\overset{q}{\mathcal{B}}_{(2)}&=\overset{q}{b}+ \overset{q+1}{b_1}\xi^1 +  \overset{q+1}{b_2}\xi^2+  \overset{q+2}{b}\xi^1\xi^2.
	\end{align*}
	Their exterior product is a generalized $(p+q)$-form of type $N=2$,
	\begin{equation}\label{N=2:wedge}
		\overset{p}{\mathcal{A}}_{(2)} \underline{\wedge} \overset{q}{\mathcal{B}}_{(2)}=	\overset{p+q}{\gamma}+ \overset{p+q+1}{\gamma_1} \xi^1+ \overset{p+q+1}{\gamma_2}\xi^2+ \overset{p+q+2}{\gamma}\xi^1\xi^2,
	\end{equation}
	where
	\begin{align*}
		\overset{p+q}{\gamma} = &\overset{p}{a} \wedge \overset{q}{b}, \\
		\overset{p+q+1}{\gamma_1} = &\overset{p}{a}\wedge  \overset{q+1}{b_1} +(-1)^q\overset{p+1}{a_1}\wedge \overset{q}{b}, \\
		\overset{p+q+1}{\gamma_2} =&\overset{p}{a}\wedge  \overset{q+1}{b_2} +(-1)^q\overset{p+1}{a_2}\wedge \overset{q}{b}, \\
		\overset{p+q+2}{\gamma} =&\overset{p}{a} \wedge\overset{q+2}{b}+(-1)^{q+1}  \overset{p+1}{a_1}\wedge\overset{q+1}{b_2} +(-1)^q \overset{p+1}{a_2}\wedge \overset{q+1}{b_1}  + \overset{p+2}{a}\wedge \overset{q}{b}.
	\end{align*}
	For forms of the same degree $p$, the symmetric inner product is
	\begin{align}\label{inn2}
		(\!(\overset{p}{\mathcal{A}}_{(2)}, \overset{p}{\mathcal{B}}_{(2)})\!)=\int_M \big(\overset{p}{a}\wedge * \overset{p}{b}+ \overset{p+1}{a_1}\wedge * \overset{p+1}{b_1}+ \overset{p+1}{a_2}\wedge * \overset{p+1}{b_2}
	+ \overset{p+2}{a}\wedge * \overset{p+2}{b}\big).
	\end{align}

	\subsubsection{Exterior derivatives and canonical basis}\label{edcb}
	We now introduce the exterior derivative for generalized forms. In the generalized setting ($N\geq 1$), a key difference from the ordinary case ($N=0$) is that the exterior derivative is no longer unique.  We will see that choosing a basis singles out a particular exterior derivative, a choice that is required in the generalized formulation of higher gauge theory.
	
	When $N=0$, the exterior derivative is the standard one on ordinary differential forms and is therefore uniquely determined.  For $N\geq1$, however, there exist multiple distinct exterior derivative operators \cite{2007R3,2009R}.   Each is a nilpotent differential $\underline{d}: \Lambda^p_{(N)}(M) \to \Lambda^{p+1}_{(N)}(M)$, satisfying the usual defining properties.  For any generalized forms $\overset{p}{\mathcal{A}}$ and $\overset{q}{\mathcal{B}}$, these are
	\begin{subequations}\label{rule}
		\begin{align}
			\underline{d} (\overset{p}{\mathcal{A}} \underline{\wedge}\, \overset{q}{\mathcal{B}})&=\underline{d}\, \overset{p}{\mathcal{A}} \underline{\wedge}\, \overset{q}{\mathcal{B}}+(-1)^p \overset{p}{\mathcal{A}} \underline{\wedge} \, \underline{d}\, \overset{q}{\mathcal{B}},\\ 
			\underline{d}(\overset{p}{\mathcal{A}} +\overset{p}{\mathcal{B}})&=\underline{d}\,\overset{p}{\mathcal{A}} +\underline{d}\, \overset{p}{\mathcal{B}}, \quad 
			\underline{d}^2 \overset{p}{\mathcal{A}}=0,\\
			\underline{d}f(X)&=X(f),
		\end{align}
	\end{subequations}
	where $f$ is any smooth function and $X$ any vector field  on $M$. For  fixed $N>0$ and  an admissible choice of $\underline{d}$, this yields a cochain complex over $\mathbb{R}$,
	\begin{align*}
		0 \stackrel{\underline{d}}{\longrightarrow}  \Lambda^{-N}_{(N)}(M) \stackrel{\underline{d}}{\longrightarrow} \cdots  \stackrel{\underline{d}}{\longrightarrow} \Lambda^{p}_{(N)}(M)\stackrel{\underline{d}}{\longrightarrow} \Lambda^{p+1}_{(N)}(M)\stackrel{\underline{d}}{\longrightarrow} 
		\cdots \stackrel{\underline{d}}{\longrightarrow} \Lambda^{n}_{(N)}(M)\stackrel{\underline{d}}{\longrightarrow} 0
	\end{align*}
	which generalizes the ordinary de Rham complex.  Since $\underline{d}^2$ and the graded Leibniz rule holds, each pair $\big(\Lambda_{(N)}(M), \underline{d}\big)$ is a differential graded algebra.
	
	Working in the basis representation \eqref{define-gfN}, the action of $\underline{d}$ on the $(-1)$-forms ${\xi^i}$ is determined by the usual graded differentiation rules:
	\begin{align*}
		\underline{d}\bigl(\xi^{i_1} \xi^{i_2} \cdots \xi^{i_k}\bigr) 
		&= \sum_{j=1}^{k} (-1)^{j-1}\, \xi^{i_1} \cdots \underline{d}\xi^{i_j} \cdots \xi^{i_k}, \\[4pt]
		\underline{d}\bigl(\overset{p}{a}\,\xi^i\bigr) 
		&= \bigl(d\overset{p}{a}\bigr)\,\xi^i + (-1)^p\,\overset{p}{a}\,\bigl(\underline{d}\xi^i\bigr), \\[4pt]
		\underline{d}\bigl(\xi^i\,\overset{p}{a}\bigr) 
		&= \bigl(\underline{d}\xi^i\bigr)\,\overset{p}{a} - \xi^i\,d\overset{p}{a}.
	\end{align*}
	Since $\underline{d}\xi^i$ is  a generalized $0$-form of type $N$,it admits the expansion
	\begin{align}\label{d}
		\underline{d}\xi^i=\mu^i- \nu^i_{i_1}\xi^{i_1}+\frac{1}{2!}\rho^i_{i_1 i_2}\xi^{i_1}\xi^{i_2} + \cdots \nonumber
		+ \frac{(-1)^N}{N!}\iota^i_{i_1 i_2 \dots i_N}\xi^{i_1}\cdots \xi^{i_N},
	\end{align}
	where the coefficients $\mu^i$, $\nu^i_{i_1}$, $\rho^i_{i_1 i_2}$, $\dots$, $\iota^i_{i_1 i_2 \dots i_N}$ are ordinary differential forms of degrees $0, 1, 2, \ldots, N$, respectively.
	The nilpotency condition $\underline{d}^2 \xi^i = 0$ then  imposes constraints on these coefficients, thereby defining a differential ideal. Conversely, any solution of this ideal determines an admissible operator $\underline{d}$ uniquely via \eqref{d}.

	A key result is that one can always choose a basis in which $\underline{d}\xi^i$ assumes a canonical form. For $N=2$, two canonical forms are established in Ref.~\cite{2007R3}.
	\begin{proposition}\label{1-case}
		Within any locally contractible domain, there exists a canonical basis $\xi^i$ such that $\underline{d} \xi^i = 0$.
	\end{proposition}
	
	\begin{proposition}\label{2-case}
		Within any locally domain, there exists a canonical basis $\xi^i$ such that $\underline{d} \xi^i = \delta^i_m k$ for $m=1$ or $2$ and $k\neq 0$.
	\end{proposition}
	Although Ref.~\cite{2007R3} treats explicitly only the case 
	$m=1$, the case  $m=2$ follows by an analogous argument. Note that such canonical bases are not unique; complete proofs are given in Ref.~\cite{2007R3}.
	
	The corresponding analysis for $N=1$ is immediate from the preceding discussion. As in the $N=2$ case, a canonical basis with $\underline{d}\xi=0$ is not unique, whereas a basis satisfying $\underline{d}\xi=k$ for a constant $k$ is unique.  This uniqueness ensures that, for $N=1$, the induced definitions of the Lie derivative, duality, codifferential, and Laplacian (on a manifold equipped with a metric) are unambiguously fixed relative to this canonical basis \cite{2002NR,2003R1}.
	
	The propositions above extend to generalized forms of arbitrary type $N$ \cite{2009R}. A particularly simple and useful choice, which enforces nilpotency automatically, is to work in a basis for  $\underline{d}\xi^i = k^i$, with $k^i$ constant  \cite{2003R2}. Once the derivatives of the basis elements are fixed, the action of $\underline{d}$ on any generalized $p$-form is completely determined. In particular, for $N=1$, with $\overset{p}{\mathcal{A}_{(1)}}=\overset{p}{a}+\overset{p+1}{a}\xi$,
	\begin{equation}\label{N=1:d}
		\underline{d}\overset{p}{\mathcal{A}_{(1)}}=d\overset{p}{a}+(-1)^{p+1}\overset{p+1}{a}k +d\overset{p+1}{a}\xi, \quad (k=\underline{d}\xi).
	\end{equation}
	For $N=2$, with $\overset{p}{\mathcal{A}_{(2)}}=\overset{p}{a}+\overset{p+1}{a}\xi^1+ \overset{p+1}{a'}\xi^2+ \overset{p+2}{a}\xi^1\xi^2$,
	\begin{align}\label{N=2:d}
		\underline{d}\overset{p}{\mathcal{A}_{(2)}}=&d \overset{p}{a}+ (-1)^{p+1} \big(k^1 \overset{p+1}{a}+ k^2 \overset{p+1}{a'}\big)+\big(d  \overset{p+1}{a}+(-1)^{p+1}k^2 \overset{p+2}{a}\big)\xi^1 \nonumber\\
		&+ \big(d  \overset{p+1}{a'}+(-1)^{p}k^1  \overset{p+2}{a}\big)\xi^2+d  \overset{p+2}{a} \xi^1\xi^2,
	\end{align}
	where $k^1=\underline{d}\xi^1$ and $k^2=\underline{d}\xi^2$.
	The constants $k$, $k^1$, and $k^2$ introduced here will be essential in the generalized formulation of higher gauge theory developed below.

	\subsection{Generalized gauge theory}\label{sec2.2}
	Having established the calculus of generalized forms, we now turn to a generalized version of gauge theory, which will later be extended to higher gauge theories. We review a generalized connection as a Lie algebra-valued generalized form \cite{2003R2,2013R}, in direct analogy with the ordinary case \cite{NM}. This framework gives rise to  generalized notions of curvature, Bianchi identities, and gauge transformations \cite{2001NR,2003R1,2002GLTZ}. In the present section,  we  work locally on $M$.  Background on algebra-valued differential forms is summarized in \ref{AVDDF}.
	
	\subsubsection{Generalized connections and curvatures}
	Let $G$ be a matrix Lie group with Lie algebra $\mathfrak{g}$. A \textbf{generalized $\mathfrak{g}$-valued connection 1-form}  of type $N$, denoted $\mathcal{A}_{(N)}$, is obtained from the standard notion of a connection by replacing ordinary differential forms with generalized forms of type $N$. For the lowest types one has:
	\begin{itemize}
		\item For $N = 1$:
		\begin{equation}\label{a_1}
			\mathcal{A}_{(1)}=\overset{1}{A}+ \overset{2}{A}\xi.
		\end{equation}
		\item For $N = 2$:
		\begin{equation}\label{a_2}
			\mathcal{A}_{(2)}=\overset{1}{A}+ \overset{2}{A}\xi^1+ \overset{2}{A'}\xi^2+ \overset{3}{A}\xi^1\xi^2.
		\end{equation}
	\end{itemize}
	Here $\overset{i}{A}$ ($i=1,2,3$) are $\mathfrak{g}$-valued $i$-forms, and $\overset{2}{A'}$ is an additional $\mathfrak{g}$-valued $2$-form.
	
	The associated \textbf{generalized curvature $2$-form} is defined, as in the ordinary case, by
	\begin{equation*}
		\mathcal{F}_{(N)} = \underline{d} \mathcal{A}_{(N)} + \frac{1}{2} [\mathcal{A}_{(N)}, \mathcal{A}_{(N)}].
	\end{equation*}
	For Lie algebra-valued forms, the graded commutator satisfies $[\mathcal{A}_{(N)}, \mathcal{A}_{(N)}] = 2 \mathcal{A}_{(N)} \underline{\wedge} \mathcal{A}_{(N)}$ (after antisymmetrization).
	Using the derivative and product rules from Subsection~\ref{sec2.1}, one finds the following component expressions.
	\begin{itemize}
		\item For $N=1$: 
		\begin{equation}\label{2.10}
			\mathcal{F}_{(1)}=\big(d\overset{1}{A} +\dfrac{1}{2}[ \overset{1}{A}, \overset{1}{A}]+k \overset{2}{A}\big) +\big( d \overset{2}{A} + [\overset{1}{A}, \overset{2}{A}]\big)\xi.
		\end{equation}
		\item For $N=2$:
		\begin{align}
			\mathcal{F}_{(2)}=&d\overset{1}{A} + \dfrac{1}{2}[ \overset{1}{A}, \overset{1}{A}] + k^1 \overset{2}{A}+ k^2 \overset{2}{A'}+ \big(d\overset{2}{A}+[\overset{1}{A}, \overset{2}{A}]+ k^2 \overset{3}{A}\big)\xi^1\nonumber\\
			& +\big(d\overset{2}{A'}+[\overset{1}{A}, \overset{2}{A'}]-k^1\overset{3}{A}\big)\xi^2+\big(  d\overset{3}{A} +[\overset{1}{A}, \overset{3}{A}] +[\overset{2}{A}, \overset{2}{A'}]\big)\xi^1\xi^2.
		\end{align}	
	\end{itemize}
	
	The curvature obeys the  \textbf{Bianchi identity}, $\mathbf{D} \mathcal{F}_{(N)}=0$, where $\mathbf{D}$ denotes the covariant exterior derivative associated with $\mathcal{A}_{(N)}$. Its action on a $\mathfrak{g}$-valued generalized $p$-form $\mathcal{W}_{(N)}$ is defined by
	\begin{equation}
		\mathbf{D} \mathcal{W}_{(N)} = \underline{d} \mathcal{W}_{(N)} + \mathcal{A}_{(N)} \underline{\wedge} \mathcal{W}_{(N)} + (-1)^{p+1} \mathcal{W}_{(N)} \underline{\wedge} \mathcal{A}_{(N)}.
	\end{equation}
	
	The construction of generalized connections and curvatures is dictated by the algebra of generalized forms and mirrors the standard formalism of ordinary gauge theory. In particular, a single generalized field such as $\mathcal{A}_{(N)}$ or $\mathcal{F}_{(N)}$ packages a collection of $\mathfrak{g}$-valued differential forms of different degrees into one object. This built-in multiplicity, inherited from the recursive definition of generalized forms, provides a convenient unified language for the higher gauge potentials and field strengths considered in the subsequent sections.

	\subsubsection{Generalized gauge transformations}\label{GGT} 
	In this framework, gauge transformations are parametrized by generalized 0-forms of type $N$ taking values in a Lie group $\mathbf{G}_{(N)}$, thereby extending the standard notion of gauge equivalence. Throughout we work in a fixed matrix representation, so that the group identity is identified with the unit matrix and homomorphisms can be treated straightforwardly.
	
	The groups $\mathbf{G}_{(N)}$ are  constructed recursively. Let $\textbf{G}_{(0)}=\{\pi\}$ be the group of ordinary $G$-valued 0-forms, with identity $1$, and let $\textbf{H}_{(0)}=\{\mu\}$ denote the additive abelian group of $\mathfrak{g}$-valued 1-forms. 
	The group $\mathbf{G}_{(0)}$ acts on $\mathbf{H}_{(0)}$ by the adjoint action $\pi \blacktriangleright \mu = \pi \mu \pi^{-1}$.
	
	A generalized 0-form of type $N=1$ is defined by 
	\begin{equation}\label{g-11}
		g_{(1)}=(1+ \mu\xi)\pi=\pi+ \mu\pi\xi,
	\end{equation}
	where $\pi\in\mathbf{G}_{(0)}$ and $\mu\in\mathbf{H}_{(0)}$.
	The set $\textbf{G}_{(1)}=\{g_{(1)}\}$ forms a group under multiplication. Given $g_{(1)}=(1+ \mu\xi)\pi$ and $g'_{(1)}=(1+ \mu'\xi)\pi'$, their product and inverse are
	\begin{align}
		g_{(1)}g'_{(1)}&=\big(1+ (\mu+ \pi \mu'  \pi^{-1})\xi\big)\pi \pi',\label{gg-1}\\
		g_{(1)}^{-1}&=\big(1 -\pi^{-1} \mu\pi\xi\big)\pi^{-1},\label{g-1}
	\end{align}
	with identity element $1_{(1)}=1$.
	The group $\textbf{G}_{(1)}$ is isomorphic to the semidirect product $G_{(0)} \ltimes H_{(0)}$. Its Lie algebra consists of elements $\overset{0}{l}_1= \overset{0}{\lambda}+ \overset{1}{\lambda}\xi$, where $\overset{0}{\lambda}$ and $\overset{1}{\lambda}$ are ordinary $\mathfrak{g}$-valued 0- and 1-forms, respectively.
	This group structure is analogous to that of a derived Lie group associated to a Lie crossed module \cite{Z-2021}, and this analogy motivates extending the construction to higher gauge theories.
	
	Proceeding inductively, given a Lie group $\textbf{G}_{(N-1)}=\{g_{(N-1)}\}$ of generalized 0-forms of type $(N-1)$ and an abelian group $\textbf{H}_{(N-1)}=\{h_{(N-1)}\}$ of generalized 1-forms of type $(N-1)$ equipped with an action of  $\textbf{G}_{(N-1)}$, a element of $\mathbf{G}_{(N)}$ is defined by
	\begin{equation}\label{gN}
		g_{(N)}=\big(1+ h_{(N-1)}\xi\big)g_{(N-1)}.
	\end{equation}
	The set $\mathbf{G}_{(N)}$ then forms a Lie group, with product and inverse formally given by  \eqref{gg-1} and \eqref{g-1} with  subscripts $(1)$ and $(0)$ replaced by $(N)$ and $(N-1)$, respectively.	Then for $N=2$, taking $h_{(1)} = \mu'+ \nu\xi$ (with $\mu'$ a $\mathfrak{g}$-valued 1-form and $\nu$ a 2-form) yields
	\begin{align}
		g_{(2)}=&\big(1+ \mu\xi^1+ \mu'\xi^2+ (\nu + \mu' \mu)\xi^1\xi^2\big)\pi,\label{g-2}\\
		g^{-1}_{(2)}=&\big(1-\pi^{-1}\mu\pi\xi^1 -\pi^{-1}\mu'\pi\xi^2-\pi^{-1}(\nu + \mu' \mu)\pi\xi^1\xi^2\big)\pi^{-1}.\label{gin-2}
	\end{align}
	
	Generalized gauge transformations, implemented by right-multiplication with elements of $\mathbf{G}_{(N)}$ \cite{2001NR,2013R}, retain the structural form of their ordinary counterparts. For the two lowest types one has:
	\begin{itemize}
		\item For $N=1$:
		Under transformation by $g_{(1)}$ as in \eqref{g-11},
		\begin{align}
			\mathcal{A}_{(1)}\longrightarrow &g^{-1}_{(1)}\underline{d}g_{(1)}+ g^{-1}_{(1)}\mathcal{A}_1g_{(1)}\nonumber\\
			=&\pi^{-1}d\pi + \pi^{-1}(\overset{1}{A}-k\mu)\pi+ \pi^{-1}( d\mu - k \mu\mu + \mu \overset{1}{A}+ \overset{1}{A}\mu + \overset{2}{A})\pi\xi,\label{A-g}\\
			\mathcal{F}_{(1)}\longrightarrow& g^{-1}_{(1)}\mathcal{F}_1g_{(1)}\nonumber\\
			=&\pi^{-1}(\overset{2}{F} +k \overset{2}{A})\pi+ \pi^{-1}\big(D\overset{2}{A} +(\overset{2}{F} + k \overset{2}{A})\mu- \mu(\overset{2}{F}+ k \overset{2}{A})
			\big) \pi\xi,
		\end{align}
		where $\overset{2}{F} =d\overset{1}{A}+\dfrac{1}{2}[\overset{1}{A}, \overset{1}{A}]$ and $D\overset{2}{A}=d\overset{2}{A}+[\overset{1}{A}, \overset{2}{A}]$.
		\item For $N=2$: Under transformation by  $g_{(2)}=(1+ \mu\xi^1+ \nu\xi^1\xi^2)\pi\in \textbf{G}_{(2)}$  \cite{2003R2},
		\begin{align}
			\mathcal{A}_{(2)}\longrightarrow &g^{-1}_{(2)}\underline{d}g_{(2)}+ g^{-1}_{(2)}\mathcal{A}_{(2)}g_{(2)}\nonumber\\
			=&\pi^{-1}\big(d \pi \pi^{-1}-k^1 \mu + \overset{1}{A}\big)\pi+\pi^{-1}\big(d\mu-k^1\mu\mu-k^2\nu+\overset{2}{A}+\overset{1}{A}\mu + \mu\overset{1}{A}\big)\pi \xi^1 \nonumber\\
			&+\pi^{-1}\big(k^1\nu + \overset{2}{A'}\big)\pi\xi^2+ \pi^{-1}\big(d\nu + k^1(\nu \mu -\mu\nu)+\overset{3}{A}+\overset{1}{A}\nu -\nu\overset{1}{A} +\overset{2}{A'}\mu- \mu \overset{2}{A'}\big)\pi\xi^1\xi^2,\label{A_2}\\
			\vspace{6mm}
			\mathcal{F}_{(2)}\longrightarrow &g^{-1}_{(2)}\mathcal{F}_{(2)}g_{(2)}\nonumber\\
			=& g^{-1}_{(2)}(\overset{2}{F}+ \overset{3}{F}\xi^1+ \overset{3}{F'}\xi^2+ \overset{4}{F}\xi^1\xi^2)g_{(2)}\nonumber\\
			=&\overset{2}{F}+\big( \overset{3}{F}+\overset{2}{F}\mu -\mu \overset{2}{F}\big)\xi^1+ \overset{3}{F'}\xi^2+\big( \overset{4}{F}+\overset{2}{F}\nu-\nu\overset{2}{F} +\overset{3}{F'}\mu +\mu \overset{3}{F'} \big)\xi^1\xi^2,
		\end{align}
		where the component curvatures are
		\begin{align*}
			\overset{2}{F}&=d\overset{1}{A} + \overset{1}{A} \wedge\overset{1}{A} + k^1 \overset{2}{A}+ k^2 \overset{2}{A'},\\
			\overset{3}{F}&=d\overset{2}{A}+[\overset{1}{A},\overset{2}{A}]+ k^2 \overset{3}{A},\\
			\overset{3}{F'}&= d\overset{2}{A'}+[\overset{1}{A}, \overset{2}{A'}]-k^1\overset{3}{A},\\
			\overset{4}{F}&= d\overset{3}{A} +[\overset{1}{A}, \overset{3}{A}] +[\overset{2}{A}, \overset{2}{A'}].
		\end{align*}
	\end{itemize}
	
	Thus, the generalized 1-forms \eqref{a_1} and \eqref{a_2} encode the field content of higher connections, while the groups $\mathbf{G}_{(N)}$  implement an extended gauge symmetry. The recursive formalism developed here for $N=1,2$ provides a template for the generalized treatment of the higher gauge transformations presented in the following sections.

	\section{Higher algebra- and group-valued generalized forms}\label{sec3}
	We now construct the generalized forms that carry higher algebraic structures. First, we revisit and adapt the higher algebra-valued generalized forms introduced in Ref.~\cite{SDH4}, recasting them in a basis representation suited to encoding higher gauge potentials. Second, we introduce the corresponding higher group-valued generalized forms, which are required to encode higher gauge transformations. Essential definitions of higher groups and algebras are collected in \ref{A.1} and \ref{A.2}; further details on algebra-valued forms are given in \ref{AVDDF}.
	
	\subsection{Higher algebra-valued generalized forms}\label{sec3.1}
	We consider generalized forms valued in Lie 2- and Lie 3-algebras, adapting the definitions of Ref.~\cite{SDH4} to our basis representation. On these spaces we introduce the graded bracket, exterior derivative, symmetric pairing, and inner product needed to formulate the HCS and HYM theories.
	
	\begin{definition}[Type-$N=1$ generalized forms valued in Lie 2-algebras]
		Let $\mathcal{O} = (\mathfrak{h}, \mathfrak{g}; \alpha, \triangleright)$ be a Lie 2-algebra. A type-$N=1$ generalized $p$-form with values in $\mathcal{O}$ is
		\begin{equation}\label{3.1}
			\mathcal{W} = U+ V\xi,
		\end{equation}
		where $U$ is a $\mathfrak{g}$-valued $p$-form and $V$ is an $\mathfrak{h}$-valued $(p+1)$-form.
	\end{definition}
	
	\begin{definition}[Type-$N=2$ generalized forms valued in Lie 3-algebras]
		Let $\mathcal{L}=(\mathfrak{l},\mathfrak{h}, \mathfrak{g}; \beta, \alpha,\triangleright, \{\cdot,\cdot \})$ be a Lie 3-algebra. A type-$N=2$ generalized $p$-form with values in $\mathcal{L}$ is
		\begin{equation}\label{3.2}
			\mathcal{W} = U+V\xi^1+ V'\xi^2+ W\xi^1\xi^2,
		\end{equation}
		where $U$ is a $\mathfrak{g}$-valued $p$-form, $V$ and $V'$ are  $\mathfrak{h}$-valued $(p+1)$-forms, and $W$ is an $\mathfrak{l}$-valued $(p+2)$-form.
	\end{definition}
	
	\textbf{Graded bilinear bracket.}
	We now equip the higher algebra-valued generalized forms introduced above with a graded Lie-algebraic structure.  For type-$N=1$ generalized forms  $\mathcal{W}_1 = U_1+ V_1\xi$ and $\mathcal{W}_2 = U_2+ V_2\xi$ of degrees $p$ and $q$, respectively, the \textbf{graded Lie bracket} is defined by
	\begin{equation}\label{Lie-g}
		[\mathcal{W}_1, \mathcal{W}_2]=[U_1, U_2]+\big(U_1\triangleright V_2-(-1)^{pq} U_2\triangleright V_1\big)\xi.
	\end{equation}
	For type $N=2$, let
	\begin{align*}
		\mathcal{W}_1 &= U_1+ V_1\xi^1+ V_1'\xi^2+W_1\xi^1\xi^2,\\
		\mathcal{W}_2 &= U_2+ V_2\xi^1+ V_2'\xi^2+ W_2\xi^1\xi^2,
	\end{align*}
	be generalized forms of degrees $p$ and $q$, respectively. The bracket extends to type $N=2$ as 
	\begin{align}\label{Lie-gg}
		[\mathcal{W}_1, \mathcal{W}_2] = &\; [U_1, U_2] + \bigl( U_1 \triangleright V_2 - (-1)^{pq} U_2 \triangleright V_1 \bigr) \xi^1+ \bigl( U_1 \triangleright V_2' - (-1)^{pq} U_2 \triangleright V_1' \bigr) \xi^2 \nonumber \\
		&+ \bigl( U_1 \triangleright W_2 - (-1)^{pq} U_2 \triangleright W_1 + (-1)^{q+1} \{V_1, V_2'\} - (-1)^{pq+p+1} \{V_2, V_1'\} \bigr) \xi^1 \xi^2 .
	\end{align}
	The brackets \eqref{Lie-g} and \eqref{Lie-gg} satisfy the \textbf{graded Jacobi identity}: for generalized forms $\mathcal{W}_1,\mathcal{W}_2,\mathcal{W}_3$ of degrees $p$, $q$, and $r$,
	\begin{equation*}
		[\mathcal{W}_1, [\mathcal{W}_2, \mathcal{W}_3]] = [[\mathcal{W}_1, \mathcal{W}_2], \mathcal{W}_3] + (-1)^{pq} [\mathcal{W}_2, [\mathcal{W}_1, \mathcal{W}_3]].
	\end{equation*}
	This follows directly from the above definitions together with the defining identities of the underlying higher-algebra operations.
	
	\textbf{Generalized exterior derivative.}
	To obtain a differential graded Lie algebra, we introduce a generalized exterior derivative,
	$$\underline{d}: \Lambda^p_{(N)}(M, \mathcal{K}) \to \Lambda^{p+1}_{(N)}(M, \mathcal{K}),\qquad ( \mathcal{K}=\mathcal{O}, \mathcal{L}),$$
	defined componentwise as follows.
	\begin{itemize}
		\item For type $N=1$ ($\mathcal{K}=\mathcal{O}$) and $\mathcal{W}=U+V\xi$,
		\begin{equation}\label{ged-dd1}
			\underline{d}\mathcal{W}=dU +(-1)^{p+1}k\alpha(V)+ d V\xi,
		\end{equation}
		which generalizes \eqref{N=1:d}.
		\item For type $N=2$ ($\mathcal{K}=\mathcal{L}$) and $\mathcal{W}=U+V\xi^1+V'\xi^2+W\xi^1\xi^2$,
		\begin{align}\label{ged-dd2}
			\underline{d}\mathcal{W}=&dU +(-1)^{p+1}\alpha(k^1V+k^2V')+\big( dV+ (-1)^{p+1}k^2\beta(W)\big)\xi^1+\big( d V' + (-1)^{p}k^1\beta(W)\big)\xi^2\nonumber\\
			&+ dW\xi^1\xi^2,
		\end{align}
		which extends \eqref{N=2:d}.
	\end{itemize}
	A straightforward calculation shows that $\underline{d}$ is nilpotent, $\underline{d}^2=0$, and satisfies the graded Leibniz rule
	\begin{equation*}
		\underline{d}	[\mathcal{W}_1, \mathcal{W}_2]= [\underline{d}\mathcal{W}_1, \mathcal{W}_2]+(-1)^p 	[\mathcal{W}_1, \underline{d}\mathcal{W}_2].
	\end{equation*}
	It follows that 
	$\big(\Lambda^\bullet_{(1)}(M, \mathcal{O}), [\cdot, \cdot], \underline{d}\big)$ and  $\big(\Lambda^\bullet_{(2)}(M, \mathcal{L}), [\cdot, \cdot], \underline{d}\big)$ are differential graded Lie algebras.

	\textbf{Generalized graded symmetric bilinear pairings.} 
	We introduce a graded symmetric bilinear pairing on higher algebra-valued generalized forms, which will be used to define actions for HCS theories.
	Let $\mathcal{K}$ be $\mathcal{O}$ for $N=1$ and $\mathcal{L}$ for $N=2$. 
	The \textbf{generalized pairing}
	$$\langle \!\langle -, - \rangle\!\rangle: \Lambda^p_{(N)}(M, \mathcal{K}) \times \Lambda^q_{(N)}(M, \mathcal{K}) \to \Lambda^{p+q}_{(N)}(M, \mathbb{R}),$$
	is defined explicitly as follows.
	\begin{itemize}
		\item Type $N=1$: For $\mathcal{W}_1 = U_1+ V_1\xi$ and $\mathcal{W}_2 = U_2+ V_2\xi$,
		\begin{equation}\label{1-s}
			\langle \!\langle\mathcal{W}_1, \mathcal{W}_2\rangle\!\rangle=\langle U_1, V_2\rangle_{\mathfrak{g}, \mathfrak{h}}+(-1)^{pq}\langle U_2, V_1\rangle_{\mathfrak{g}, \mathfrak{h}},
		\end{equation}
		where $\langle-,-\rangle_{\mathfrak{g},\mathfrak{h}}$ is a bilinear form on the Lie 2-algebra $\mathcal{O}$ (see \ref{A.1}).
		\item Type $N=2$: For $\mathcal{W}_1 = U_1+ V_1\xi^1+ V_1'\xi^2+W_1\xi^1\xi^2$ and $\mathcal{W}_2 = U_2+ V_2\xi^1+ V_2'\xi^2+ W_2\xi^1\xi^2$, 
		\begin{equation}\label{2-s}
			\begin{aligned}
				\langle \!\langle\mathcal{W}_1, \mathcal{W}_2\rangle\!\rangle=&\langle U_1, W_2\rangle_{\mathfrak{g}, \mathfrak{l}}+(-1)^{pq}\langle U_2, W_1\rangle_{\mathfrak{g}, \mathfrak{l}}
				-k^1\langle V_1, V'_2\rangle_{\mathfrak{h}}-(-1)^{pq}k^2\langle V_2, V'_1\rangle_{\mathfrak{h}},
			\end{aligned}
		\end{equation}
		where $\langle -, - \rangle_{\mathfrak{g}, \mathfrak{l}}$ and $\langle -, - \rangle_{\mathfrak{h}}$ are invariant forms of the Lie 3-algebra $\mathcal{L}$ (see \ref{A.2}).
	\end{itemize}
	Besides, these pairings are bilinear and graded symmetric, namely,
	\begin{equation*}
		\langle \!\langle\mathcal{W}_1, \mathcal{W}_2\rangle\!\rangle=(-1)^{pq}\langle \!\langle\mathcal{W}_2, \mathcal{W}_1\rangle\!\rangle.
	\end{equation*}
	
	\textbf{Symmetric inner products for higher algebra-valued generalized forms.}
	We now define symmetric inner products on these spaces, which will be used to construct actions for HYM theories. 
	\begin{itemize}
		\item Type $N=1$: For  $\mathcal{W}_1 = U_1+ V_1\xi$ and $\mathcal{W}_2 = U_2+ V_2$,  the inner product is
		\begin{equation}\label{g-inn1}
			(\!(\mathcal{W}_1, \mathcal{W}_2)\!)=\int_M \big(\langle U_1, * U_2\rangle_{\mathfrak{g}}+ \langle V_1, * V_2\rangle_{\mathfrak{h}}\big).
		\end{equation}
		which follows from  \eqref{inn1}.
		\item Type $N=2$: For $\mathcal{W}_1 = U_1+ V_1\xi^1+ V_1'\xi^2+W_1\xi^1\xi^2$ and $\mathcal{W}_2 = U_2+ V_2\xi^1+ V_2'\xi^2+ W_2\xi^1\xi^2$, 
		\begin{align}\label{g-inn2}
			(\!(\mathcal{W}_1, \mathcal{W}_2)\!)=\int_M \big(\langle U_1, * U_2\rangle_{\mathfrak{g}}+ \langle V_1, * V_2\rangle_{\mathfrak{h}}
			+ \langle V'_1, * V'_2\rangle_{\mathfrak{h}}+\langle W_1, * W_2\rangle_{\mathfrak{l}}\big),
		\end{align}
		as induced by \eqref{inn2}.
	\end{itemize}

	\subsection{Higher group-valued generalized forms}\label{sec3.2}
	In this subsection, we extend the construction of ordinary gauge groups developed in Subsection~\ref{GGT} to Lie $2$- and $3$-groups, and we introduce the corresponding higher Maurer–Cartan forms. The resulting higher group-valued generalized $0$-forms provide the symmetry parameters for higher gauge transformations, which will be detailed in Section~\ref{sec5}.
	
	\subsubsection{Type $N=1$: Lie 2-group-valued generalized 0-forms}
	This part focuses on the symmetry parameters of $2$-gauge theory. We define the group of Lie $2$-group-valued generalized $0$-forms, study its Maurer–Cartan form, and establish properties of its adjoint action on the generalized form space $\Lambda^\bullet_{(1)}(M,\mathcal{O})$.
	
	\begin{definition}[\textbf{Lie $2$-group-valued generalized $0$-forms}]
		Let $\mathbf{O}=(H,G;\bar{\alpha},\bar{\triangleright})$ be a Lie $2$-group with corresponding Lie $2$-algebra $\mathcal{O}=(\mathfrak{h},\mathfrak{g};\alpha,\triangleright)$. A $\mathbf{O}$-valued generalized $0$-form is
		\begin{equation}\label{tra-ged}
			\mathcal{G}=(1+\phi\xi)g
		\end{equation}
		with $\phi \in \Lambda^1(M, \mathfrak{h})$ and $g \in C^{\infty}(M, G)$. 
	\end{definition}
	
	By definition, $\mathcal{G}=(1+\phi\xi)g$ has components in two groups: the group
	$\mathbf{G}_{(0)}=\{g\}$ of $G$-valued functions and the additive group
	$\mathbf{H}_{(0)}=\{\phi\}$ of $\mathfrak{h}$-valued $1$-forms. The action of
	$\mathbf{G}_{(0)}$ on $\mathbf{H}_{(0)}$ is induced by the Lie $2$-group action
	$\triangleright: G\times \mathfrak{h}\to \mathfrak{h}$ specified in \ref{A.1}.
	Consequently, the set $\mathbf{G}_{(1)}=\{\mathcal{G}\}$ carries a natural group structure. For
	$\mathcal{G}=(1+\phi\xi)g$ and $\mathcal{G}'=(1+\phi'\xi)g'$, the product and inverse are
	\begin{align*}
		\mathcal{G}\mathcal{G'}&=\big(1+ (\phi + g\triangleright \phi')\xi\big)gg',\\
		\mathcal{G}^{-1}&=\big(1 -g^{-1}\triangleright \phi\xi\big)g^{-1},
	\end{align*}
	which are induced by \eqref{gg-1} and \eqref{g-1}. The identity element is $1$. This shows that
	$\mathcal{G}$ is the direct higher analogue of \eqref{g-11} via the recursive scheme \eqref{gN}.

	In analogy with the ordinary Maurer–Cartan form, we define the \textbf{2-Maurer–Cartan form} (generalized left fundamental 1-form) associated with $\mathcal{G}$ by
	\begin{align}\label{2MC}
		l_2=\mathcal{G}^{-1}\underline{d}\mathcal{G}=g^{-1}dg - kg^{-1}\alpha(\phi)g+\big( g^{-1}\triangleright (d\phi -k \phi \phi)\big)\xi,
	\end{align}
	where $k$ is the constant entering the definition of $\underline{d}$ for type $N=1$. This expression extends the earlier quantity $g^{-1}_{(1)}\underline{d}g_{(1)}$ in \eqref{A-g} to the Lie $2$-group-valued setting.
	
	The form $l_2$ obeys a higher analogue of the classical Maurer–Cartan equation.
	\begin{theorem}
		The 2-Maurer–Cartan form $l_2$ 
		satisfies the \textbf{2-Maurer–Cartan equation}
		\begin{align*}
			\underline{d}l_2+\dfrac{1}{2}[l_2, l_2]=0.
		\end{align*}
	\end{theorem}
	\begin{proof}
		Starting from \eqref{ged-dd1} and \eqref{gY2}, one obtains
		\begin{align}\label{4.7}
			\underline{d}l_2
			=& dg^{-1}dg - k\,dg^{-1}\alpha(\phi)g + k\,g^{-1}\alpha(\phi)dg - k^{2} g^{-1}\triangleright \alpha(\phi \phi) \nonumber \\
			& + \Bigl( dg^{-1}\triangleright \bigl(d\phi - k \phi \phi\bigr) + k\,g^{-1}\triangleright [\phi, d\phi] \Bigr)\xi .
		\end{align}
		On the other hand, using the bracket in \eqref{Lie-g} we obtain
		\begin{align}\label{4.8}
			[l_2,l_2] &=\big[g^{-1}dg - k g^{-1}\alpha(\phi)g,\; g^{-1}dg - k g^{-1}\alpha(\phi)g\big] + 2\big(g^{-1}dg - k g^{-1}\alpha(\phi)g\big)\!\triangleright\!\big(g^{-1}\!\triangleright\!(d\phi-k\phi\phi)\big)\xi \nonumber\\
			&=2g^{-1}dg\,g^{-1}dg 
			- 2k\big[g^{-1}\alpha(\phi)g,\;g^{-1}dg\big] 
		+ (k)^{2}\big[g^{-1}\alpha(\phi)g,\;g^{-1}\alpha(\phi)g\big] \nonumber\\
			&\quad+\Big(2(g^{-1}dg)\!\triangleright\!\big(g^{-1}\!\triangleright\!(d\phi-k\phi\phi)\big) -2k\big(g^{-1}\alpha(\phi)g\big)\!\triangleright\!\big(g^{-1}\!\triangleright\!(d\phi-k\phi\phi)\big)\Big)\xi .
		\end{align}
		The following identities simplify these expressions:
		\begin{align}
			&g^{-1}dg\,g^{-1}dg = -dg^{-1}dg, \label{4.9} \\
			&\bigl[g^{-1}\alpha(\phi)g,\; g^{-1}dg\bigr] = g^{-1}\alpha(\phi)dg - dg^{-1}\alpha(\phi)g, \label{4.10} \\
			&	\bigl[g^{-1}\alpha(\phi)g,\; g^{-1}\alpha(\phi)g\bigr] = 2 g^{-1}\alpha(\phi\phi)g. \label{4.11}
		\end{align}
		Moreover, using \eqref{YyY'} and \eqref{A.6} we have
		\begin{align}
			&(g^{-1}dg)\triangleright(g^{-1}\triangleright ( d\phi -k \phi\phi))
			=-dg^{-1}\triangleright ( d\phi -k \phi\phi),\nonumber\\
			&(g^{-1}\alpha(\phi)g )\triangleright\big(g^{-1}\triangleright ( d\phi -k \phi\phi)\big)
			=g^{-1}\triangleright [\phi, d\phi].\label{4.13}
		\end{align}
		Substituting Eqs.~\eqref{4.9}--\eqref{4.13} into \eqref{4.8} and combining  the result to \eqref{4.7} yields
		\begin{equation*}
			\underline{d}\ell_2 + \dfrac{1}{2}[\ell_2, \ell_2] = 0,
		\end{equation*}
		which completes the proof. 
	\end{proof}
	
	Given a $\mathfrak{g}$-valued 1-form $A\in\Lambda^1(M,\mathfrak{g})$ and an $\mathfrak{h}$-valued 2-form $B\in\Lambda^2(M,\mathfrak{h})$, we introduce
	\begin{equation*}
		\textbf{H}_{(1)}=\{h_{(1)}=A+B\xi\}\subset \Lambda^\bullet_{(1)}(M, \mathcal{O}).
	\end{equation*} 
	The adjoint action of $\mathbf{G}_{(1)}$ on $\mathbf{H}_{(1)}$ is obtained by restricting the natural action of $\mathbf{G}_{(1)}$ on the full space $\Lambda^{\bullet}_{(1)}(M,\mathcal{O})$. For a general element $\mathcal{W}=U+V\xi \in \Lambda^{p}_{(1)}(M,\mathcal{O})$, this action is
	\begin{align}\label{ad}
		\operatorname{\textbf{Ad}}_{\mathcal{G}}\mathcal{W}=\operatorname{Ad}_g U+ \big(g\triangleright V - \operatorname{Ad}_g U \triangleright\phi\big)\xi.
	\end{align}

	\begin{theorem}\label{4.2}
		Let $\mathcal{W}_1=U_1+V_1\xi$ and $\mathcal{W}_2=U_2+V_2\xi$ be generalized forms of degrees $p$ and $q$, respectively, in $\Lambda^{\bullet}_{(1)}(M,\mathcal{O})$. 
		Then the adjoint action \eqref{ad} commutes with the graded Lie bracket \eqref{Lie-g}, namely,
		\begin{equation*}
			[\operatorname{\textbf{Ad}}_{\mathcal{G}}\mathcal{W}_1, \operatorname{\textbf{Ad}}_{\mathcal{G}}\mathcal{W}_2 ]=\operatorname{\textbf{Ad}}_{\mathcal{G}}[\mathcal{W}_1, \mathcal{W}_2].
		\end{equation*}
	\end{theorem}
	\begin{proof}
		A direct computation gives
		\begin{align*}
			&[\operatorname{Ad}_{\mathcal{G}}\mathcal{W}_1, \operatorname{Ad}_{\mathcal{G}}\mathcal{W}_2 ]\nonumber\\
			=&[\operatorname{Ad}_g U_1+ (g \triangleright V_1-\operatorname{Ad}_g U_1 \triangleright\phi)\xi, \operatorname{Ad}_g U_2+( g \triangleright V_2-\operatorname{Ad}_g U_2 \triangleright\phi)\xi]\nonumber\\
			=&[\operatorname{Ad}_g U_1, \operatorname{Ad}_g U_2]+ \big(\operatorname{Ad}_g U_1 \triangleright(g \triangleright V_2 -\operatorname{Ad}_g U_2 \triangleright\phi)\nonumber\\
			& - (-1)^{pq} \operatorname{Ad}_g U_2 \triangleright(g \triangleright V_1 - \operatorname{Ad}_g U_1 \triangleright\phi)\big)\xi\nonumber\\
			=&\operatorname{Ad}_g[U_1, U_2]+\big( g\triangleright (U_1\triangleright V_2- (-1)^{pq}U_2 \triangleright V_1)-\operatorname{Ad}_g[U_1, U_2]\triangleright\phi\big)\xi\nonumber\\
			=&\operatorname{Ad}_{\mathcal{G}}[\mathcal{W}_1, \mathcal{W}_2].
		\end{align*}
	\end{proof}
	
	\begin{theorem}\label{4.3}
		Let $\mathcal{W}_1$ and $\mathcal{W}_2$ be generalized forms in \\ $\Lambda^{\bullet}_{(1)}(M,\mathcal{O})$. Then the type $N=1$ generalized bilinear form $\langle \!\langle -, - \rangle\!\rangle$ defined in \eqref{1-s} is invariant under the $\mathbf{G}_{(1)}$-action, i.e.
		\begin{equation}
			\langle \!\langle \operatorname{\textbf{Ad}}_{\mathcal{G}}\mathcal{W}_1, \operatorname{\textbf{Ad}}_{\mathcal{G}}\mathcal{W}_2\rangle\!\rangle=\langle \!\langle \mathcal{W}_1,\mathcal{W}_2\rangle\!\rangle.
		\end{equation}
	\end{theorem}
	\begin{proof}
		Let $\mathcal{W}_1=U_1+V_1\xi\in \Lambda^{p}_{(1)}(M,\mathcal{O})$ and $\mathcal{W}_2=U_2+V_2\xi \in \Lambda^{q}_{(1)}(M,\mathcal{O})$.
		Using \eqref{1-s} and \eqref{ad}, we compute
		\begin{align*}
			&\langle \!\langle \operatorname{\textbf{Ad}}_{\mathcal{G}}\mathcal{W}_1, \operatorname{\textbf{Ad}}_{\mathcal{G}}\mathcal{W}_2\rangle\!\rangle\\
			=&
			\langle \!\langle \operatorname{Ad}_g U_1+\big( g\triangleright V_1 - \operatorname{Ad}_g U_1 \triangleright\phi\big)\xi, \operatorname{Ad}_g U_2+\big( g\triangleright V_2 - \operatorname{Ad}_g U_2 \triangleright\phi\big)\xi \rangle\!\rangle\\
			=&\langle \operatorname{Ad}_g U_1, g\triangleright V_2 - \operatorname{Ad}_g U_2 \triangleright\phi\rangle_{\mathfrak{g}, \mathfrak{h}} +(-1)^{pq}\langle \operatorname{Ad}_g U_2, g\triangleright V_1- \operatorname{Ad}_g U_1 \triangleright\phi\rangle_{\mathfrak{g}, \mathfrak{h}}.
		\end{align*}
		By the $G$-invariance of $\langle-,-\rangle_{\mathfrak{g},\mathfrak{h}}$ in \eqref{gin} and identity \eqref{A.7}, this becomes
		\begin{align*}
			\langle \!\langle \operatorname{\textbf{Ad}}_{\mathcal{G}}\mathcal{W}_1, \operatorname{\textbf{Ad}}_{\mathcal{G}}\mathcal{W}_2\rangle\!\rangle
			=\langle  U_1,  V_2 -  U_2 \triangleright(g^{-1}\triangleright\phi)\rangle_{\mathfrak{g}, \mathfrak{h}}+(-1)^{pq}\langle  U_2, V_1 -  U_1 \triangleright(g^{-1}\triangleright\phi)\rangle_{\mathfrak{g}, \mathfrak{h}}.
		\end{align*}
		The terms involving $g^{-1}\triangleright\phi$ cancel by the symmetry relation \eqref{XXY}, which implies
		\begin{align*}
			\langle U_1, U_2\triangleright (g^{-1}\triangleright \phi)\rangle_{\mathfrak{g}, \mathfrak{h}}=(-1)^{pq+1}\langle U_2, U_1\triangleright (g^{-1}\triangleright \phi)\rangle_{\mathfrak{g}, \mathfrak{h}}.
		\end{align*}
		Therefore,
		\begin{align*}
			\langle \!\langle \operatorname{\textbf{Ad}}_{\mathcal{G}}\mathcal{W}_1, \operatorname{\textbf{Ad}}_{\mathcal{G}}\mathcal{W}_2\rangle\!\rangle&=\langle U_1, V_2\rangle_{\mathfrak{g}, \mathfrak{h}}+(-1)^{pq}\langle U_2, V_1\rangle_{\mathfrak{g}, \mathfrak{h}}\\
			&=\langle \!\langle\mathcal{W}_1, \mathcal{W}_2\rangle\!\rangle,
		\end{align*}
		which proves the claim.
		
	\end{proof}

	\subsubsection{Type $N=2$: Lie $3$-group-valued generalized $0$-forms}
	We next consider the recursive scheme for type $N=2$. This leads to the definition of Lie $3$-group-valued generalized $0$-forms, which serve as the gauge symmetry parameters of the $3$-gauge theory.
	
	\begin{definition}[\textbf{Lie $3$-group-valued generalized $0$-forms}]
		Let $\textbf{L}=(L, H, G;\beta, \alpha, \triangleright, \left\{\cdot,\cdot\right\})$ be a Lie $3$-group with associated Lie $3$-algebra $\mathcal{L}=(\mathfrak{l}, \mathfrak{h}, \mathfrak{g};\beta, \alpha, \triangleright, \left\{\cdot,\cdot\right\})$.  An $\mathbf{L}$-valued generalized $0$-form is an element of the form
		\begin{equation}\label{G-2}
			\mathcal{G}=(1+ \phi_1\xi^1+ \phi_2\xi^2+ \psi\xi^1\xi^2)g
		\end{equation}
		with $\phi_1, \phi_2\in \Lambda^1(M, \mathfrak{h})$, $\psi \in \Lambda^2(M, \mathfrak{l})$ and $g\in C^{\infty}(M, G)$. 
	\end{definition}
	
	Expression \eqref{G-2} is a higher analogue of the type $N=2$ group element in \eqref{g-2}. Defining $\mathbf{G}_{(2)}=\{\mathcal{G}\}$, this set inherits a group structure from the product $g_{(2)}g'_{(2)}$ described in Subsection~\ref{GGT}. Explicitly, for $\mathcal{G}$ and
	\[
	\mathcal{G}'=\bigl(1+\phi'_1\xi^1+\phi'_2\xi^2+\psi'\,\xi^1\xi^2\bigr)g',
	\]
	the product is
	\begin{equation*}
		\mathcal{G}\mathcal{G'}=\big(1+( \phi_1 + g\triangleright \phi'_1)\xi^1+( \phi_2 + g\triangleright \phi'_2)\xi^2+ (\psi + g\triangleright \psi')\xi^1\xi^2\big)gg'.
	\end{equation*}
	The inverse, following directly from \eqref{gin-2}, is 
	\begin{equation}
		\mathcal{G}^{-1}=\big(1 -g^{-1}\triangleright\phi_1\xi^1 -g^{-1}\triangleright\phi_2\xi^2 -g^{-1}\triangleright \psi\xi^1\xi^2\big)g^{-1}.
	\end{equation}

	We now specialize to the simplifying case considered in {2003R2}, where
	\begin{equation}
		\mathcal{G}=(1+ \phi\xi^1+ \psi\xi^1\xi^2)g.
	\end{equation}
	In analogy with the type $N=1$ case, we define the $3$-Maurer–Cartan form by
	\begin{align}\label{3MC}
		l_3=&\mathcal{G}^{-1}\underline{d}\mathcal{G}\nonumber\\
		=&g^{-1}dg - k^1g^{-1}\alpha(\phi)g\nonumber\\
		&+ g^{-1}\triangleright \big(d\phi -k^1 \phi \phi -k^2\beta(\psi)\big)\xi^1 +k^1g^{-1}\triangleright \beta(\psi)\xi^2\nonumber\\
		&+ g^{-1}\triangleright\big( d\psi -k^1\phi \triangleright' \psi\big)\xi^1\xi^2.
	\end{align}
	
	\begin{theorem}
		The form $l_3$ satisfies the \textbf{3-Maurer–Cartan equation} 
		\begin{equation}\label{3MC-eq}
			\underline{d}l_3+\dfrac{1}{2}[l_3, l_3]=0,
		\end{equation}
		provided the constants in the nilpotent operator $\underline{d}$ obey $k^1=k^2$.
	\end{theorem}
	
	\begin{proof}
		A direct computation establishes \eqref{3MC-eq}.  Applying the exterior derivative $\underline{d}$ defined in \eqref{ged-dd2} to $l_3$ gives
		\begin{align}
			\underline{d}l_3=&dg^{-1}dg-k^1dg^{-1}\alpha(\phi)g-k^1g^{-1}\alpha(d\phi)g\nonumber\\
			& +k^1g^{-1}\alpha(\phi)dg+k^1\alpha\big(g^{-1}\triangleright (d\phi - k^1 \phi\phi-k^2\beta(\psi))\big)\nonumber\\
			&+k^2\alpha\big(k^1g^{-1}\triangleright \beta(\psi)\big)+\Big(dg^{-1}\triangleright\big(d\phi-k^1\phi\phi-k^2\beta(\psi)\big)\nonumber\\
			&+k^1\beta\big(g^{-1}\triangleright(d\psi-k^1\phi\triangleright' \psi)\big)+g^{-1}\triangleright\big(-k^1d\phi\phi\nonumber\\
			&+k^1\phi d\phi-k^2\beta(d\psi)\big)\Big)\xi^1\nonumber\\
			&+\Big(k^1dg^{-1}\triangleright \beta(\psi)+k^1g^{-1}\triangleright\beta(d\psi)\nonumber\\
			&-k^2\beta\big(g^{-1}\triangleright (d\psi-k^1\phi\triangleright' \psi)\big)\Big)\xi^2\nonumber\\
			&+\Big(dg^{-1}\triangleright\big(d\psi-k^1\phi \triangleright' \psi\big)-k^1g^{-1}\triangleright \big(d\phi\triangleright' \psi\big)\nonumber\\
			&+k^1g^{-1}\triangleright \big(\phi \triangleright' d\psi\big)
			\Big)\xi^1\xi^2.
		\end{align}
		Using \eqref{gY2} and $\alpha\circ\beta=0$, we have
		\begin{align}\label{alpha-simp}
			\alpha(g^{-1}\triangleright d\phi)=g^{-1} \alpha(d\phi)g,\ \ \alpha\big(k^1g^{-1}\triangleright \beta(\psi)\big)=0.
		\end{align}
		Moreover, applying \eqref{gZ} and imposing $k^1=k^2$ yields
		\begin{align}\label{beta-simp}
			\beta\big(g^{-1}\triangleright d\psi\big)=g^{-1}\triangleright \beta(d\psi).
		\end{align}
		Substituting \eqref{alpha-simp} and \eqref{beta-simp} gives the simplified form
		\begin{align}\label{4.26}
			\underline{d}l_3=&dg^{-1}dg-k^1dg^{-1}\alpha(\phi)g+k^1g^{-1}\alpha(\phi)dg \nonumber\\
			& +k^1\alpha\big(g^{-1}\triangleright (- k^1 \phi\phi)\big)+\Big(dg^{-1}\triangleright\big(d\phi-k^1\phi\phi\nonumber\\
			&-k^2\beta(\psi)\big)+k^1\beta\big(g^{-1}\triangleright(-k^1\phi\triangleright' \psi)\big)\nonumber\\
			& +g^{-1}\triangleright\big(-k^1d\phi\phi+k^1\phi d\phi\big)\Big)\xi^1\nonumber\\
			&+\Big(k^1dg^{-1}\triangleright \beta(\psi)-k^2\beta\big(g^{-1}\triangleright (-k^1\phi\triangleright' \psi)\big)\Big)\xi^2\nonumber\\
			&+\Big(dg^{-1}\triangleright\big(d\psi-k^1\phi \triangleright' \psi\big)-k^1g^{-1}\triangleright \big(d\phi\triangleright' \psi\big)\nonumber\\
			&+k^1g^{-1}\triangleright \big(\phi \triangleright' d\psi\big)
			\Big)\xi^1\xi^2.
		\end{align}
		
		We now compute the bracket term using \eqref{Lie-gg}:
		\begin{align}\label{l3l3}
			\dfrac{1}{2}[l_3, l_3]
			=& \dfrac{1}{2}\big[g^{-1}dg-k^1g^{-1}\alpha(\phi)g, g^{-1}dg-k^1g^{-1}\alpha(\phi)g\big]\nonumber\\
			&+\big(g^{-1}dg-k^1g^{-1}\alpha(\phi)g\big)\triangleright\big(g^{-1}\triangleright(d\phi-k^1\phi\phi\nonumber\\
			&-k^2\beta(\psi))\big)\xi^1\nonumber\\
			&+\big(g^{-1}dg-k^1g^{-1}\alpha(\phi)g\big)\triangleright\big(k^1g^{-1}\triangleright \beta(\psi)\big)\xi^2\nonumber\\
			&+\Big(\big(g^{-1}dg-k^1g^{-1}\alpha(\phi)g\big)\triangleright\big(g^{-1}\triangleright (d\psi-k^1\phi\triangleright'\psi)\big)\nonumber\\
			&+\big\{g^{-1}\triangleright\big(d\phi-k^1\phi\phi-k^2\beta(\psi)\big),k^1g^{-1}\triangleright\beta(\psi)\big\}
			\Big)\xi^1\xi^2.
		\end{align}
		
		The following steps simplify each component of \eqref{l3l3}.   \textbf{First component:}
		Using $g^{-1}dgg^{-1}=-dg^{-1}$ and\\ $g^{-1}\alpha(\phi)\alpha(\phi)g=\alpha\big(g^{-1}\triangleright(\phi\phi)\big)$, we obtain
		\begin{align}\label{4.28}
			&\dfrac{1}{2}\big[g^{-1}dg-k^1g^{-1}\alpha(\phi)g, g^{-1}dg-k^1g^{-1}\alpha(\phi)g\big]\nonumber\\
			=&-dg^{-1}dg-k^1g^{-1}\alpha(\phi)dg +k^1 dg^{-1}\alpha(\phi)g+(k^1)^2\alpha\big(g^{-1}\triangleright(\phi\phi)\big).
		\end{align}
		\textbf{Second and third components:} Applying \eqref{A.7} and the identities
		\begin{align*}
			\alpha(\phi)\triangleright d\phi&=[\phi, d\phi],\\ \alpha(\phi)\triangleright(\phi\phi)&=[\phi, \phi\phi]=0,\\ \alpha(\phi)\triangleright\beta(\psi)&=\beta(\phi\triangleright'\psi)=[\phi, \beta(\psi)],
		\end{align*}
		yields
		\begin{align}\label{4.29}
			&\big(g^{-1}dg-k^1g^{-1}\alpha(\phi)g\big)\triangleright\big(g^{-1}\triangleright(d\phi-k^1\phi\phi-k^2\beta(\psi))\big)\nonumber\\
			=&g^{-1}\triangleright\Big(\big(dgg^{-1}-k^1\alpha(\phi)\big)\triangleright\big(d\phi-k^1\phi\phi-k^2\beta(\psi)\big)\Big)\nonumber\\
			=& g^{-1}\triangleright\Big((dgg^{-1})\triangleright d\phi-k^1[\phi, d\phi]-k^1(dgg^{-1})\triangleright(\phi\phi)-k^2(dgg^{-1})\triangleright \beta(\psi)+k^1k^2\beta(\phi\triangleright'\psi)
			\Big)
		\end{align}
		and 
		\begin{align}\label{4.30}
			\big(g^{-1}dg-k^1g^{-1}\alpha(\phi)g\big)\triangleright\big(k^1g^{-1}\triangleright \beta(\psi)\big)
			=g^{-1}\triangleright \Big((dgg^{-1})\triangleright k^1\beta(\psi)-(k^1)^2\beta(\phi\triangleright'\psi)
			\Big).
		\end{align}
		\textbf{Fourth component:} Using \eqref{YZZ}, \eqref{A.17} and \eqref{B.7}, we have
		\begin{align}
			\alpha(\phi)\triangleright d\psi&=\phi \triangleright' d\psi,\\
			\{d\phi, \beta(\psi)\}&=d\phi \triangleright'\psi,\\
			\{ \beta(\psi), \beta(\psi)\}&=[\psi, \psi]=0,
		\end{align}
		and also
		\begin{align}
			\alpha(\phi)\triangleright (\phi \triangleright' \psi)&=\phi \triangleright' (\phi \triangleright' \psi)=(\phi\phi)\triangleright' \psi
			\nonumber\\
			&=-\{\beta(\psi), \phi\phi\}=\{\phi\phi, \beta(\psi)\}.
		\end{align}
		Consequently,
		\begin{align}\label{4.33}
			&\big(g^{-1}dg-k^1g^{-1}\alpha(\phi)g\big)\triangleright\big(g^{-1}\triangleright (d\psi-k^1\phi\triangleright'\psi)\big)\nonumber\\&+\big\{g^{-1}\triangleright\big(d\phi -k^1\phi\phi-k^2\beta(\psi)\big),k^1g^{-1}\triangleright\beta(\psi)\big\}\nonumber\\
			=&g^{-1}\triangleright \Big( \big(dgg^{-1}-k^1\alpha(\phi)\big)\triangleright\big(d\psi-k^1\phi\triangleright'\psi\big)\Big)\nonumber\\&+g^{-1}\triangleright\{d\phi-k^1\phi\phi-k^2\beta(\psi), k^1\beta(\psi)	\}\nonumber\\
			=&g^{-1}\triangleright k^1 (d\phi\triangleright' \psi-\phi\triangleright' d\psi),
		\end{align}
		where the last step uses \eqref{A.24} and \eqref{gyy}.

		Finally, substituting expressions \eqref{4.28}, \eqref{4.29}, \eqref{4.30}, and \eqref{4.33} into \eqref{l3l3}, and adding the result to \eqref{4.26}, one checks that all terms cancel pairwise, proving \eqref{3MC-eq}. 
	\end{proof}
	
	In particular, the generalized Maurer–Cartan form $l_3$ encodes the corresponding flatness condition for the Lie $3$-group $\mathbf{L}$.
	
	Let $A\in \Lambda^1(M, \mathfrak{g})$, $B_1, B_2\in \Lambda^2(M, \mathfrak{h})$ and $C\in \Lambda^3(M, \mathfrak{l})$. Define
	\begin{equation*}
		\textbf{H}_{(2)}=\{h_{(2)}=A+B_1\xi^1+B_2\xi^2+C\xi^1\xi^2\}\subset \Lambda^\bullet_{(2)}(M, \mathcal{L}).
	\end{equation*} 
	The adjoint action of  $\textbf{G}_{(2)}$ on $\textbf{H}_{(2)}$ is defined as the restriction of its natural action on  the full space  $\Lambda^\bullet_{(2)}(M, \mathcal{L})$.  For a general element $	\mathcal{W} = U+V\xi^1+ V'\xi^2+ W\xi^1\xi^2$ as in \eqref{3.2},  this adjoint action is given explicitly by
	\begin{align}\label{ad2}
		\operatorname{\textbf{Ad}}_{\mathcal{G}}\mathcal{W}=&\operatorname{Ad}_g U+\big( g\triangleright V - (\operatorname{Ad}_g U) \triangleright\phi\big)\xi^1+g\triangleright V'\xi^2\nonumber\\
		&+  \big(g\triangleright W -(\operatorname{Ad}_gU)\triangleright \psi -\{g\triangleright V', \phi\}+\{\phi, g\triangleright V'\}\big)\xi^1\xi^2.
	\end{align}

	We find that the present adjoint action does not satisfy the properties established in Theorems~\ref{4.2} and~\ref{4.3}. Rather, it should be viewed as a direct extension of the construction for type $N=2$. These limitations are inherent to the current framework. Nonetheless, pursuing this higher generalization is instructive. In the following chapters, we therefore continue along this path and derive the corresponding $3$-gauge transformations, which constitute a necessary step in the systematic exploration of higher gauge structures.


	\section{Higher gauge theories in generalized formulations}\label{sec4}
	Based on the higher generalized forms developed in Section~\ref{sec3}, we now present generalized formulations of $2$- and $3$-gauge theories. This framework clarifies the structural correspondence between higher and ordinary gauge theories, and thereby facilitates the transfer of both conceptual insights and technical tools from the ordinary to the higher setting.
	\subsection{Type $N=1$: 2-connections and 2-gauge transformations}\label{sec4.1}
	The $2$-gauge theory is the categorification of ordinary gauge theory, obtained by replacing the structure Lie group with a Lie $2$-group.
	In analogy to a Lie group $G$ giving rise to a connection $A$, a Lie $2$-group  $(H, G; \bar{\alpha}, \bar{\triangleright})$ carries a \textbf{2-connection} $(A, B)$, consisting of a standard  $\mathfrak{g}$-valued $1$-form $A \in \Lambda^1(M, \mathfrak{g})$ and an $\mathfrak{h}$-valued $2$-form $B \in \Lambda^2(M, \mathfrak{h})$.
	
	The $2$-connection defines a \textbf{fake $2$-curvature} $\Omega_1$ and a \textbf{$2$-curvature} $\Omega_2$ by
	\begin{subequations}\label{2-curvature}
		\begin{align}
			\Omega_1&=dA + \dfrac{1}{2}[A, A]-\alpha(B)\in  \Lambda^2(M, \mathfrak{g}),\\
			\Omega_2&=dB+A\triangleright B\in  \Lambda^3(M, \mathfrak{h}),
		\end{align}
	\end{subequations}
	which satisfy \textbf{2-Bianchi identities}
	\begin{subequations}\label{2-BI}
		\begin{align}
			&d\Omega_1+[A, \Omega_1]+\alpha(\Omega_2)=0,\\
			&d\Omega_2 + A \triangleright \Omega_2 - \Omega_1 \triangleright B=0.
		\end{align}
	\end{subequations}
	The adjective ``fake'' refers to the additional term $\alpha(B)$ in $\Omega_1$, which distinguishes it from the ordinary curvature \cite{Baez.2010}. 
	Throughout, we work with a strict Lie $2$-algebra; in particular, we do not need to impose the fake-flatness condition $\Omega_1=0$, which is typically required in the weak Lie $2$-algebra setting \cite{HKCS}.

	In the language of Subsection~\ref{sec3.1}, the gauge field $(A,B)$ can be packaged into a type $N=1$ generalized $1$-form with values in the Lie $2$-algebra \eqref{3.1},
	\begin{equation}\label{conn1}
		\mathcal{A}=A+ B\xi,
	\end{equation}
	which we call a \textbf{Lie 2-algebra-valued generalized connection}.
	For two such forms $\mathcal{A}_i = A_i + B_i\xi$ ($i=1,2$), the graded bracket \eqref{Lie-g} yields
	\begin{align}\label{12A}
		[\mathcal{A}_1, \mathcal{A}_2]=[A_1, A_2]+ \big(A_1\triangleright B_2+ A_2\triangleright B_1\big)\xi.
	\end{align}
	Choosing $k=-1$ in the generalized exterior derivative \eqref{ged-dd1}, we obtain
	\begin{equation}\label{12d}
		\underline{d}\mathcal{A}=dA -\alpha(B)+ d B\xi.
	\end{equation}
	The associated \textbf{generalized curvature} is then
	\begin{equation*}
		\mathcal{F}=\underline{d}\mathcal{A}+ \dfrac{1}{2}[\mathcal{A}, \mathcal{A}]=\Omega_1+\Omega_2\xi.
	\end{equation*}
	Thus, the higher curvatures are encoded in the single generalized form $\mathcal{F}$, whose algebraic structure mirrors that of the ordinary curvature.
	By construction, $\mathcal{F}$ is a type $N=1$ generalized $2$-form, and its components reproduce the fake $2$-curvature and the $2$-curvature. 
	Moreover, $\mathcal{F}$ satisfies the \textbf{generalized Bianchi identity}
	\begin{align*}
		\underline{d}\mathcal{F}+[\mathcal{A}, \mathcal{F}]
		=&d\Omega_1+[A, \Omega_1]+\alpha(\Omega_2)+\big( d\Omega_2 + A \triangleright \Omega_2 - \Omega_1 \triangleright B\big)\xi\nonumber\\
		=&0,
	\end{align*}
	which is equivalent to the $2$-Bianchi identities \eqref{2-BI}.
	
	A $2$-gauge transformation can be encoded by a Lie $2$-group-valued generalized $0$-form
	\begin{equation*}
		\mathcal{G}=(1+\phi\xi)g \in \mathbf{G}_{(1)},
	\end{equation*}
	with $\phi \in \Lambda^1(M, \mathfrak{h})$ and $g \in C^{\infty}(M, G)$. 
	Using \eqref{ad}, the adjoint action on the generalized connection is
	\begin{equation*}
		\textbf{Ad}_{\mathcal{G}}\mathcal{A}=\operatorname{Ad}_g A+\Big(g\triangleright B - \operatorname{Ad}_g A \triangleright\phi\Big)\xi.
	\end{equation*}
	Replacing $\mathcal{G}$ by $ \mathcal{G}^{-1}$ (i.e., $g \to g^{-1}$, $\phi \to -g^{-1} \triangleright \phi$) yields the inverse adjoint action.
	The \textbf{generalized gauge transformation} is then
	\begin{align*}
		\mathcal{A}'=\textbf{Ad}_{\mathcal{G}^{-1}}\mathcal{A} + \mathcal{G}^{-1}	\underline{d}\mathcal{G}.
	\end{align*}
	where the second term is the $2$-Maurer–Cartan form \eqref{2MC}. This formula is formally identical to the usual gauge transformation in ordinary gauge theory. Consequently, the generalized curvature transforms covariantly:
	\begin{equation}
		\mathcal{F}'=\textbf{Ad}_{\mathcal{G}^{-1}}\mathcal{F}.
	\end{equation}
	In components, the induced transformation of the $2$-connection $(A,B)$ is
	\begin{align*}
		A'&=\operatorname{Ad}_{g^{-1}} (A + dg g^{-1} +\alpha(\phi)),\\
		B'&= g^{-1}\triangleright (B +A \triangleright\phi + d \phi - k \phi \phi).
	\end{align*}
	and the curvatures transform as
	\begin{equation*}
		\Omega'_1=\operatorname{Ad}_{g^{-1}} \Omega_1,\ \ \ 
		\Omega'_2=g^{-1}\triangleright(\Omega_2+\Omega_1\triangleright \phi),
	\end{equation*}
	in agreement with earlier results \cite{Z-2021}. 
	
	Besides, the transformation admits the equivalent form
	\begin{align*}
		\mathcal{A}'=&\operatorname{Ad}_{\mathcal{G}}\mathcal{A} -d\mathcal{G}\mathcal{G}^{-1}\nonumber\\
		=&\operatorname{Ad}_g A - dgg^{-1}- \alpha(\phi)+\big( g\triangleright B - d \phi - \phi \phi \nonumber\\
		&- (\operatorname{Ad}_g A - dgg^{-1}- \alpha(\phi))\triangleright\phi \big)\xi,
	\end{align*}
	so that, in components,
	\begin{subequations}\label{Ag1}
		\begin{align}
			A'&=\operatorname{Ad}_g A - dgg^{-1}- \alpha(\phi),\\
			B'&=g\triangleright B - d \phi - \phi \phi -  A'\triangleright\phi.
		\end{align}
	\end{subequations}
	These expressions coincide with those in Refs.~\cite{K.W,SZ,Zucchini-2020-1,Zucchini-2020-2}.

	\subsection{Type $N=2$: 3-connections and 3-gauge transformations}\label{sec4.2}
	We now extend the generalized formulation to $3$-gauge theory \cite{JFM,doi:10.1063/1.4870640}. While the algebraic structure becomes richer, it still admits a unified description in the language of generalized forms.
	
	A Lie $3$-group (more precisely, a Lie $2$-crossed module) $(L, H, G;\bar{\beta}, \bar{\alpha}, \bar{\triangleright}, \left\{\cdot,\cdot\right\})$ gives rise to a 3-connection $(A, B, C)$, consisting of  $A \in \Lambda^1(M, \mathfrak{g})$, $B \in \Lambda^2(M, \mathfrak{h})$, and $C \in \Lambda^3(M, \mathfrak{l})$. The associated \textbf{fake 1-curvature} $\Omega_1$, \textbf{fake 2-curvature} $\Omega_2$, and \textbf{3-curvature} $\Omega_3$ are defined by
	\begin{subequations}\label{3-cu}
		\begin{align}
			\Omega_1&=dA+\dfrac{1}{2}[A, A]-\alpha(B),\\ \Omega_2&=dB+A\triangleright B-\beta(C),\\
			\Omega_3&=dC + A \triangleright C +\{B, B\},
		\end{align}
	\end{subequations}
	and they satisfy the \textbf{$3$-Bianchi identities}
	\begin{subequations}\label{3BI}
		\begin{align}
			&d\Omega_1+[A, \Omega_1]+\alpha(\Omega_2)=0,\\
			&d\Omega_2 + A \triangleright \Omega_2 -\Omega_1 \triangleright B + \beta(\Omega_3)=0,\\
			&d\Omega_3 + A \triangleright \Omega_3 - \Omega_1\triangleright C -\{B, \Omega_2\}-\{\Omega_2, B\}=0.
		\end{align}
	\end{subequations}
	
	Within the generalized framework, the gauge field $(A,B,C)$ is encoded by a type $N=2$ generalized $1$-form valued in the Lie $3$-algebra \eqref{3.2},
	\begin{equation}\label{2cnn}
		\mathcal{A}=A+ B\xi^1+ B\xi^2+ C\xi^1\xi^2,
	\end{equation}
	which we call a \textbf{Lie 3-algebra-valued generalized connection}. 
	
	For two such generalized connections $\mathcal{A}_i = A_i+ B_i\xi^1+ B_i\xi^2+ C_i\xi^1\xi^2$ ($i = 1, 2$), the graded bracket \eqref{Lie-gg} yields
	\begin{align}\label{AA3}
		[\mathcal{A}_1, \mathcal{A}_2]=&[A_1, A_2]+\big(A_1\triangleright B_2+ A_2\triangleright B_1\big)\xi^1\nonumber\\
		&+\big( A_1\triangleright B_2'+ A_2\triangleright B_1'\big)\xi^2 \nonumber
		\\
		&+\big(A_1\triangleright C_2 +A_2\triangleright C_1+\{B_1, B_2'\}+\{B_2, B_1'\}\big)\xi^1\xi^2.
	\end{align}
	Choosing $k^{1}=0$ and $k^{2}=-1$ in the generalized exterior derivative \eqref{ged-dd2}, we obtain
	\begin{equation}\label{d3}
		\underline{d}\mathcal{A}=dA -\alpha(B)+\big( d B -\beta(C)\big)\xi^1+ dB\xi^2+ dC\xi^1\xi^2.
	\end{equation}
	
	The corresponding \textbf{generalized curvature}  is defined by
	\begin{equation*}
		\mathcal{F}=\underline{d}\mathcal{A}+ \dfrac{1}{2}[\mathcal{A}, \mathcal{A}],
	\end{equation*}
	which again has the same formal shape as in ordinary gauge theory. By construction, $\mathcal{F}$ is a type $N=2$ generalized $2$-form. Its components are
	\begin{equation}\label{3-cur}
		\mathcal{F}= \Omega_1+ \Omega_2\xi^1+ \big(\Omega_2+\beta(C)\big)\xi^2+ \Omega_3\xi^1\xi^2.
	\end{equation} 
	It satisfies the \textbf{generalized Bianchi identity},
	\begin{equation}
		\underline{d}\mathcal{F}+[\mathcal{A}, \mathcal{F}]=0,
	\end{equation}
	which reproduces the $3$-Bianchi identities \eqref{3BI}.
	
	A $3$-gauge transformation is encoded by a Lie $3$-group-valued generalized $0$-form
	\begin{equation}
		\mathcal{G}=(1+ \phi\xi^1+ \psi\xi^1\xi^2)g\in \mathbf{G}_{(2)},
	\end{equation}
	with $\phi\in \Lambda^1(M, \mathfrak{h})$, $\psi \in \Lambda^2(M, \mathfrak{l})$ and $g\in C^{\infty}(M, G)$. 
	From \eqref{ad2} we obtain
	\begin{align*}
		\operatorname{Ad}_{\mathcal{G}}\mathcal{A}=&\operatorname{Ad}_g A+\Big( g\triangleright B - \operatorname{Ad}_g A \triangleright\phi\Big)\xi^1+ g\triangleright B\xi^2\nonumber\\
		&+\Big( g\triangleright C -\operatorname{Ad}_gA\triangleright \psi -\{g\triangleright B, \phi\}+\{\phi, g\triangleright B\}\Big)\xi^1\xi^2.
	\end{align*}
	Replacing $\mathcal{G}$ by its inverse (so that $g \to g^{-1}$, $\phi \to -g^{-1} \triangleright \phi$, $\psi \to -g^{-1} \triangleright \psi$) yields the inverse adjoint action
	\begin{align}\label{ad-2}
		\operatorname{Ad}_{\mathcal{G}^{-1}}\mathcal{A}=&\operatorname{Ad}_{g^{-1}} A+ g^{-1}\triangleright (B+ A \triangleright\phi)\xi^1 + g^{-1}\triangleright B\xi^2\nonumber\\
		&+g^{-1}\triangleright \big(C+ A\triangleright \psi +\{B, \phi\}-\{\phi, B\}\big)\xi^1\xi^2.
	\end{align}
	The \textbf{generalized gauge transformation} for $\mathcal{A}$ is then
	\begin{align}\label{5.22}
		\mathcal{A}'=&\operatorname{Ad}_{\mathcal{G}^{-1}}\mathcal{A} + \mathcal{G}^{-1}\underline{d}\mathcal{G}\nonumber\\
		=&\operatorname{Ad}_{g^{-1}} A+g^{-1}dg - (k^1 + t)g^{-1}\alpha(\phi)g\nonumber\\
		&+ g^{-1}\triangleright \big(B + A \triangleright\phi +d\phi -(k^1 + t) \phi \phi \nonumber\\
		&-k^2\beta(\psi)\big)\xi^1+ g^{-1}\triangleright\big(B+(k^1 + t) \beta(\psi)\big)\xi^2\nonumber\\
		&+ g^{-1}\triangleright\big(C+ A \triangleright \psi+ \{B, \phi\}- \{\phi, B\}  \nonumber\\
		&+ d\psi -(k^1 + t)\phi \triangleright' \psi\big)\xi^1\xi^2,
	\end{align}
	where the second term involves the $3$-Maurer–Cartan form \eqref{3MC} with the replacement $k^1 \to k^1 + t$. 
	
	The constants $k^1$ and $k^2$ are fixed by choosing a canonical basis $\{\bar{\xi}^i\}$  satisfying $\underline{d} \xi^i=k^i$ (see Subsection~\ref{edcb}). Taking
	\begin{equation*}
		k^1 =0,\ \  k^2 =t= -1,
	\end{equation*}
	we obtain the explicit $3$-gauge transformations,
	\begin{align*}
		A'&=\operatorname{Ad}_{g^{-1}} A+g^{-1}dg +g^{-1}\alpha(\phi)g,\\
		B'&= g^{-1}\triangleright \big(B + A \triangleright\phi +d\phi +\phi \phi +\beta(\psi)\big),\\
		C'&=g^{-1}\triangleright\big(C+ A \triangleright \psi + \{B, \phi\}- \{\phi, B\} + d\psi +\phi \triangleright' \psi\big).
	\end{align*}
	Note that the $\xi^{2}$-component in \eqref{5.22} does not contribute to these final formulae.   Moreover, the parameter $t$ is introduced to ensure consistency between the values of $k^1$ and $k^2$ in the transformation and those in the generalized curvature.

	\section{Applications}\label{sec5}
	Building on the theoretical framework developed in the preceding sections, we now turn to concrete applications. We construct HCS and HYM theories using higher generalized forms. Our aims are twofold: to show explicitly how these higher gauge theories are built, and to clarify both their connections to and their differences from standard theories formulated with ordinary differential forms.
	
	\subsection{Higher Chern–Simons theory}\label{sec5.1}
	
	\subsubsection{4D 2-Chern–Simons  theory }
	Guided by the Maurer–Cartan formulation of HCS actions \cite{Jurco2019,ASBV}, we construct these theories using generalized forms, in close analogy to the method employed for topological invariants in Refs.~\cite{FGAO,HC}. 
	
	Just as ordinary Chern–Simons theory \cite{DSF,TP} is defined only on odd-dimensional manifolds, the $2$CS theory is naturally formulated on $4$-manifolds. For a $2$-connection $(A,B)$, we define the corresponding $2$CS form using generalized forms, in close analogy with the standard Chern–Simons $3$-form.
	\begin{definition}
		For  a type $N=1$ generalized connection $\mathcal{A}=A+ B\xi$ (see \eqref{conn1}), the \textbf{2CS 4-form} is defined by
		\begin{equation}\label{2CS-4}	\operatorname{CS}(\mathcal{A}):=\operatorname{CS}_4(A, B):=\langle \!\langle \mathcal{A}, \underline{d}\mathcal{A}+\dfrac{1}{3}[\mathcal{A}, \mathcal{A}]\rangle\!\rangle.
		\end{equation}
	\end{definition}
	
	\begin{lemma}
		The 2CS form \eqref{2CS-4} expands in components as
		\begin{equation}\label{eq:CS4_expanded}
			\operatorname{CS}_4(A, B)=\langle 2F-\alpha(B), B\rangle_{\mathfrak{g}, \mathfrak{h}} -d\langle A, B\rangle_{\mathfrak{g}, \mathfrak{h}},
		\end{equation}
		where $F=dA+\dfrac{1}{2}[A, A]$.
	\end{lemma}
	\begin{proof}
		Using \eqref{12A} and \eqref{12d}, we compute
		\begin{align*}
			\operatorname{CS}_4(A, B)
			=&\langle \!\langle A+B\xi, dA+\dfrac{1}{3}[A, A]-\alpha(B)+\big( dB+ \dfrac{2}{3}A\triangleright B\big)\xi\rangle\!\rangle\\
			=&\langle A, dB+ \dfrac{2}{3}A\triangleright B\rangle_{\mathfrak{g}, \mathfrak{h}}+\langle dA+\dfrac{1}{3}[A, A]-\alpha(B), B\rangle_{\mathfrak{g}, \mathfrak{h}}.
		\end{align*}
		By  $\mathfrak{g}$-invariance \eqref{XXY}, 
		\begin{equation*}
			\langle A,\, A \triangleright B \rangle_{\mathfrak{g}, \mathfrak{h}} = \langle [A, A],\, B \rangle_{\mathfrak{g}, \mathfrak{h}}.
		\end{equation*}
		Combining this with 
		\begin{equation*}
			d\langle A, B \rangle_{\mathfrak{g}, \mathfrak{h}} = \langle dA, B \rangle_{\mathfrak{g}, \mathfrak{h}} - \langle A, dB \rangle_{\mathfrak{g}, \mathfrak{h}},
		\end{equation*}
		we obtain
		\begin{equation*}
			\mathrm{CS}_4(A, B) = \langle 2dA + [A, A] - \alpha(B), B \rangle_{\mathfrak{g}, \mathfrak{h}} - d\langle A, B \rangle_{\mathfrak{g}, \mathfrak{h}}.
		\end{equation*}
		Since $2F=2dA+[A,A]$, this is exactly \eqref{eq:CS4_expanded}.
		
	\end{proof}
	
	The resulting expression differs from the $2$CS form in Ref.~\cite{SDH4} by the exact term $d\langle A,B\rangle_{\mathfrak{g},\mathfrak{h}}$. Consequently, the two definitions yield identical actions on any manifold without boundary.
	
	In the same spirit, we construct a \textbf{2-Chern form} from the generalized curvature $\mathcal{F}=\Omega_1+ \Omega_2\xi$:
	\begin{equation}\label{P-5}
		P(\mathcal{F}):=P_5(\Omega_1, \Omega_2):=\langle \!\langle \mathcal{F}, \mathcal{F}\rangle\!\rangle.
	\end{equation}
	Expanding in components gives
	\begin{equation*}
		P_5(\Omega_1, \Omega_2)=2\langle \Omega_1, \Omega_2\rangle_{\mathfrak{g}, \mathfrak{h}}.
	\end{equation*}
	
	\begin{proposition}
		The exterior derivative of the $2$CS $4$-form equals the $2$-Chern $5$-form:
		\begin{equation}\label{2Chern-form}
			P_5(\Omega_1, \Omega_2)=d\operatorname{CS}_4(A, B).
		\end{equation}
	\end{proposition}
	\begin{proof}
		Starting from the component expression \eqref{eq:CS4_expanded},
		\begin{align*}
			d\operatorname{CS}_4(A, B)=&d\langle 2F-\alpha(B), B\rangle_{\mathfrak{g}, \mathfrak{h}}\nonumber\\
			=&\langle 2d\Omega_1+\alpha(dB), B\rangle_{\mathfrak{g}, \mathfrak{h}}+\langle 2\Omega_1+\alpha(B), dB\rangle_{\mathfrak{g}, \mathfrak{h}}.
		\end{align*}
		Using the $2$-Bianchi identities \eqref{2-BI} to eliminate $d\Omega_1$, this becomes
		\begin{align}\label{eq:dCS4}
			d\operatorname{CS}_4(A, B)=&2\langle[\Omega_1, A]-\alpha(\Omega_2), B\rangle_{\mathfrak{g}, \mathfrak{h}}+2\langle \Omega_1+\alpha(B), \Omega_2-A\triangleright B\rangle_{\mathfrak{g}, \mathfrak{h}}.
		\end{align}
		From \eqref{symp} and \eqref{XXY}, we have
		\begin{align*}
			\langle [\Omega_1, A], B\rangle_{\mathfrak{g}, \mathfrak{h}}&=\langle \Omega_1, A\triangleright B\rangle_{\mathfrak{g}, \mathfrak{h}},\\
			\langle \alpha(B), A\triangleright B\rangle_{\mathfrak{g}, \mathfrak{h}}&=\langle \alpha(A\triangleright B), B\rangle_{\mathfrak{g}, \mathfrak{h}}.
		\end{align*}
		The Peiffer identity \eqref{YyY'} implies
		\begin{align*}
			\alpha(A\triangleright B)&=[A, \alpha(B)], \\
			\alpha(B)\triangleright B&=[B, B]=0,
		\end{align*}
		and hence
		\begin{align*}
			\langle \alpha(\Omega_2), B\rangle_{\mathfrak{g}, \mathfrak{h}}&=\langle \alpha(B), \Omega_2\rangle_{\mathfrak{g}, \mathfrak{h}},\\
			\langle \alpha(B), A\triangleright B\rangle_{\mathfrak{g}, \mathfrak{h}}&=\langle [A, \alpha(B)], B\rangle_{\mathfrak{g}, \mathfrak{h}}\\
			&=\langle A, \alpha(B)\triangleright B\rangle_{\mathfrak{g}, \mathfrak{h}}\\
			&=0.
		\end{align*}
		Substituting these relations into \eqref{eq:dCS4} yields
		\begin{align*}
			d\operatorname{CS}_4(A, B)=2\langle \Omega_1, \Omega_2\rangle_{\mathfrak{g}, \mathfrak{h}}=	P_5(\Omega_1, \Omega_2).
		\end{align*}
		
	\end{proof}
	
	Equation \eqref{2Chern-form} is equivalent to the compact generalized identity
	\begin{equation*}
		\langle \!\langle \mathcal{F}, \mathcal{F}\rangle\!\rangle=\underline{d}\langle \!\langle \mathcal{A}, \underline{d}\mathcal{A}+\dfrac{1}{3}[\mathcal{A}, \mathcal{A}]\rangle\!\rangle,
	\end{equation*}
	which mirrors the structure of the ordinary Chern–Weil theorem. The 2-Chern form $P_5(\Omega_1, \Omega_2)$ further satisfies essential properties such as  \textbf{2-gauge invariance} and a \textbf{2-Chern–Weil theorem}.  A complete discussion can be found in Ref.~\cite{SDH4}.
	
	We now examine the behaviour of the 2CS 4-form under 2-gauge transformations. A key result from Refs.~\cite{ASBV,HC-2024} will be used in the analysis.
	\begin{proposition}
		Under the 2-gauge transformation \eqref{Ag1}, the 2CS 4-form \eqref{eq:CS4_expanded} transforms as
		\begin{align*}
			\operatorname{CS}_4(A', B')=&\operatorname{CS}_4(A, B)-d\big(\langle gAg^{-1}, F(\phi)\rangle_{\mathfrak{g}, \mathfrak{h}}+ \langle \alpha(\phi), d\phi+\dfrac{1}{3}[\phi, \phi]\rangle_{\mathfrak{g}, \mathfrak{h}}\nonumber\\
			&-\langle dgg^{-1}+\alpha(\phi), g\triangleright B+F(\phi)\rangle_{\mathfrak{g}, \mathfrak{h}}\big).
		\end{align*}
	\end{proposition}
	
	The proposition shows that the gauge variation of the 2CS form is a total derivative. Hence, the  non invariance of
	2CS theory is  holographic, i.e.,  it reduces to a boundary term  \cite{Z-2021}. This contrasts with ordinary Chern–Simons theory, whose gauge variation produces the well-known Wess--Zumino--Witten term. 
	
	On a closed 4-manifold $M_4$, the 2CS action is defined by
	\begin{equation}\label{2CS}
		S_{2CS}=\int_{M_4} \operatorname{CS}(\mathcal{A})=\int_{M_4} \langle 2F(A) - \alpha(B), B \rangle_{\mathfrak{g}, \mathfrak{h}}.
	\end{equation}
	Two important special cases are worth noting:
	\begin{itemize}
		\item If the crossed module is \emph{skeletal} ($\alpha=0$), the action reduces to 4D BF theory.
		\vspace{-1mm}
		\item If the crossed module is \emph{trivial} ($\alpha=\mathrm{id}$), the action reproduces, in particular, the 4D BF-BB theory, which is conjectured to coincide with the Crane--Yetter--Broda topological field theory \cite{JCB-1996,MGFGCG,FGPT}.
	\end{itemize}

	\subsubsection{5D 3-Chern–Simons  theory}
	Following the same construction pattern as for the 2CS theory, we now build a 3CS theory based on the Lie 3-algebra-valued generalized connection. The theory is naturally defined on 5-dimensional manifolds and resides within the framework of 3-gauge theory. For a 3-connection $(A, B, C)$, we  define the corresponding 3CS form.
	\begin{definition}
		For  a type $N=2$ generalized connection
		$\mathcal{A}=A+ B\xi^1+ B\xi^2+ C\xi^1\xi^2$ (see \eqref{2cnn}), the \textbf{3CS 5-form} is defined by
		\begin{equation}\label{3CS-5}
			\operatorname{CS}(\mathcal{A}):=\operatorname{CS}_5(A, B, C):=\langle \!\langle \mathcal{A}, \underline{d}\mathcal{A}+\dfrac{1}{3}[\mathcal{A}, \mathcal{A}]\rangle\!\rangle.
		\end{equation}
	\end{definition}
	
	\begin{lemma}
		The 3CS form \eqref{3CS-5} expands in components as
		\begin{align}\label{eq:CSexpanded}
			\operatorname{CS}_5(A, B, C)=\langle 2F-\alpha(B), C\rangle_{\mathfrak{g}, \mathfrak{l}}+\langle B, \Omega_2\rangle_{\mathfrak{h}}-d\langle A, C\rangle_{\mathfrak{g}, \mathfrak{l}}
		\end{align}
		where $F=dA+\dfrac{1}{2}[A, A]$.
	\end{lemma}
	\begin{proof}
		Using \eqref{AA3} and \eqref{d3}, we compute
		\begin{align}\label{CS_5}
			&	\operatorname{CS}_5(A, B, C)\nonumber\\
			=&\langle \!\langle A+ B\xi^1+ B\xi^2+ C\xi^1\xi^2, dA -\alpha(B)+\big( d B -\beta(C)\big)\xi^1\nonumber\\
			&+ dB\xi^2+ dC\xi^1\xi^2+\dfrac{1}{3}\big([A, A]+ 2A\triangleright B\xi^1+ 2A\triangleright B\xi^2\nonumber\\
			&+ (2A\triangleright C +2\{B, B\})\xi^1\xi^2\big)\rangle\!\rangle\nonumber\\
			=&\langle A, dC + \dfrac{2}{3}(A\triangleright C +\{B, B\})\rangle_{\mathfrak{g}, \mathfrak{l}} +\langle dA-\alpha(B)\nonumber\\
			&+\dfrac{1}{3}[A, A], C\rangle_{\mathfrak{g}, \mathfrak{l}}+\langle B, dB-\beta(C)+\dfrac{2}{3}A\triangleright B\rangle_{\mathfrak{h}}.
		\end{align}
		From \eqref{XZ} and \eqref{XYY}, we obtain 
		\begin{align*}
			\langle A, A\triangleright C\rangle_{\mathfrak{g}, \mathfrak{l}}=\langle [A, A], C\rangle_{\mathfrak{g}, \mathfrak{l}},\ \ \langle A, \{B, B\}\rangle_{\mathfrak{g}, \mathfrak{l}}=\dfrac{1}{2}\langle B, A\triangleright B\rangle_{\mathfrak{h}},
		\end{align*}
		and the identity
		\begin{align*}
			\langle A, dC\rangle_{\mathfrak{g}, \mathfrak{l}}=\langle dA, C\rangle_{\mathfrak{g}, \mathfrak{l}}-d\langle A, C\rangle_{\mathfrak{g}, \mathfrak{l}}.
		\end{align*}
		Substituting these relations into \eqref{CS_5} yields
		\begin{align*}
			\operatorname{CS}_5(A, B, C)=\langle 2F-\alpha(B), C\rangle_{\mathfrak{g}, \mathfrak{l}}+\langle B, \Omega_2\rangle_{\mathfrak{h}}-d\langle A, C\rangle_{\mathfrak{g}, \mathfrak{l}},
		\end{align*}
		where $F=dA+\frac12[A,A]$.
		
	\end{proof}
	
	As in the 2CS case, our expression for the 3CS form differs from that of Ref.~\cite{SDH4} only by the addition of the exact term $d\langle A,C\rangle_{\mathfrak{g},\mathfrak{l}}$. Consequently, the two prescriptions define the same action functional on manifolds without boundary.
	
	The generalized curvature associated with $\mathcal{A}=A+ B\xi^1+ B\xi^2+ C\xi^1\xi^2$ is
	\begin{equation*}
		\mathcal{F}= \Omega_1+ \Omega_2\xi^1+ \big(\Omega_2+\beta(C)\big)\xi^2+ \Omega_3\xi^1\xi^2.
	\end{equation*} 
	In order to construct a Chern-type form expressed directly in terms of the curvature components of the 3-connection, we instead consider the simplified generalized curvature
	\begin{equation*}
		\bar{\mathcal{F}}= \Omega_1+ \Omega_2\xi^1+ \Omega_2\xi^2+ \Omega_3\xi^1\xi^2,
	\end{equation*} 
	whose components are exactly $(\Omega_1,\Omega_2,\Omega_3)$. We then define the associated \textbf{3-Chern 6-form} by
	\begin{align*}
		P(\bar{\mathcal{F}}):&=P_6(\Omega_1, \Omega_2, \Omega_3):=\langle \!\langle 	\bar{\mathcal{F}}, 	\bar{\mathcal{F}}\rangle\!\rangle\nonumber\\
		&=2\langle \Omega_1, \Omega_3\rangle_{\mathfrak{g}, \mathfrak{l}}+\langle \Omega_2, \Omega_2\rangle_{\mathfrak{h}}.
	\end{align*}
	
	\begin{proposition}
		The exterior derivative of the 3CS 5-form is the 3-Chern 6-form, 
		\begin{equation}\label{3Chern-3CS}
			P_6(\Omega_1, \Omega_2, \Omega_3)=d\operatorname{CS}_5(A, B, C).
		\end{equation}
	\end{proposition}
	\begin{proof}
		Starting from the component expression \eqref{eq:CSexpanded}, we compute
		\begin{align}\label{dcs5}
			d\operatorname{CS}_5(A, B, C)=& d\langle 2\Omega_1+\alpha(B), C\rangle_{\mathfrak{g}, \mathfrak{l}}+d\langle B, \Omega_2\rangle_{\mathfrak{h}}\nonumber\\
			=&2\langle d\Omega_1, C\rangle_{\mathfrak{g}, \mathfrak{l}}+\langle \alpha(dB), C\rangle_{\mathfrak{g}, \mathfrak{l}}+\langle 2\Omega_1+\alpha(B), dC\rangle_{\mathfrak{g}, \mathfrak{l}}\nonumber\\
			&+\langle dB, \Omega_2\rangle_{\mathfrak{h}}+\langle B, d\Omega_2\rangle_{\mathfrak{h}}.
		\end{align}
		Inserting the 3-Bianchi identities \eqref{3BI}, we obtain
		\begin{align}\label{dcs51}
			d\operatorname{CS}_5(A, B, C)
			=&2\langle [\Omega_1, A]-\alpha(\Omega_2), C\rangle_{\mathfrak{g}, \mathfrak{l}}+\langle \beta(C), dB\rangle_{\mathfrak{g}, \mathfrak{l}} +2\langle \Omega_1,dC\rangle_{\mathfrak{g}, \mathfrak{l}}\nonumber\\
			&+\langle \alpha(B), \Omega_3-\{B, B\}-A\triangleright C\rangle_{\mathfrak{g}, \mathfrak{l}}+\langle\Omega_2, dB\rangle_{\mathfrak{h}}\nonumber\\
			&+\langle B, \Omega_1\triangleright B-A\triangleright \Omega_2-\beta(\Omega_3)\rangle_{\mathfrak{h}},
		\end{align}
		where we  used $\langle \alpha(dB), C\rangle_{\mathfrak{g}, \mathfrak{l}}=\langle \beta(C), dB\rangle_{\mathfrak{h}}$. 
		
		The invariance properties \eqref{XZ} and \eqref{YZ} imply
		\begin{align*}
			\bigl\langle\alpha(B),\,\Omega_3\bigr\rangle_{\mathfrak{g},\mathfrak{l}} &= \bigl\langle B,\,\beta(\Omega_3)\bigr\rangle_{\mathfrak{h}}, 
			\\ \bigl\langle[\Omega_1,A],\,C\bigr\rangle_{\mathfrak{g},\mathfrak{l}} &= \bigl\langle\Omega_1,\,A\triangleright C\bigr\rangle_{\mathfrak{g},\mathfrak{l}},\\[2pt]
			\bigl\langle\alpha(\Omega_2),\,C\bigr\rangle_{\mathfrak{g},\mathfrak{l}} &= \bigl\langle\Omega_2,\,\beta(C)\bigr\rangle_{\mathfrak{h}}, 
			\\
			\bigl\langle\beta(C),\,\beta(C)\bigr\rangle_{\mathfrak{h}} &= 0,\\[2pt]
			\bigl\langle\beta(C),\,A\triangleright B\bigr\rangle_{\mathfrak{h}} &= -\bigl\langle\alpha(B),\,A\triangleright C\bigr\rangle_{\mathfrak{g},\mathfrak{l}}, 
			\\
			\bigl\langle B,\,A\triangleright\Omega_2\bigr\rangle_{\mathfrak{h}} &= -\bigl\langle\Omega_2,\,A\triangleright B\bigr\rangle_{\mathfrak{h}}.
		\end{align*}
		Moreover, by  \eqref{XYY} we have
		\begin{align*}
			\langle \alpha(B), \{B, B\}\rangle_{\mathfrak{g}, \mathfrak{l}}&=2\langle B, \alpha(B)\triangleright B\rangle_{\mathfrak{h}}=0,\\
			\langle B, \Omega_1\triangleright B\rangle_{\mathfrak{h}}&=2\langle \Omega_1, \{B, B\}\rangle_{\mathfrak{g}, \mathfrak{l}}.
		\end{align*}
		Substituting these identities into \eqref{dcs51} and simplifying yields
		\begin{align*}
			d\operatorname{CS}_5(A, B, C)=2\langle \Omega_1, \Omega_3\rangle_{\mathfrak{g}, \mathfrak{l}}+\langle \Omega_2, \Omega_2\rangle_{\mathfrak{h}}.
		\end{align*}
		
	\end{proof}
	
	Equation \eqref{3Chern-3CS} admits an equivalent formulation in terms of generalized forms:
	\begin{equation*}
		\langle \!\langle \bar{\mathcal{F}}, \bar{\mathcal{F}}\rangle\!\rangle=\underline{d}\langle \!\langle \mathcal{A}, \underline{d}\mathcal{A}+\dfrac{1}{3}[\mathcal{A}, \mathcal{A}]\rangle\!\rangle.
	\end{equation*}
	Although this identity is formally analogous to the corresponding statements in ordinary Chern–Simons theory and 2CS theory, its interpretation in the present generalized framework is different, since $\bar{\mathcal{F}}$ is not the curvature of $\mathcal{A}$. Nevertheless, \eqref{3Chern-3CS} appears as a natural higher analogue of the Chern–Weil theorem. The 3-Chern form  is likewise invariant under 3-gauge transformations. A detailed discussion can be found in Ref.~\cite{SDH4}.
	
	On a closed 5-manifold $M_5$, the 3CS action is defined by
	\begin{align}\label{3CS}
		S_{3CS}=&\int_{M_5} \operatorname{CS}(\mathcal{A})\nonumber\\
		=&\int_{M_5}\big(\langle 2F-\alpha(B), C\rangle_{\mathfrak{g}, \mathfrak{l}}+\langle B, \Omega_2\rangle_{\mathfrak{h}}\big).
	\end{align}

	In principle, one may also investigate the transformation of the  3CS form under 3-gauge transformations. At present, however, the algebraic complexity of the required manipulations has so far precluded an explicit closed-form result. In view of the interesting holographic behaviour found for the 2CS theory, a systematic analysis of this question remains an interesting direction for future work.

	
	\subsection{Higher Yang–Mills theory}\label{sec5.2}
	Yang–Mills theory is the non-abelian generalization of Maxwell electrodynamics. Accordingly, just as Maxwell theory admits a higher extension within the framework of higher gauge theory \cite{Pfeiffer,Henneaux},  a parallel family of HYM theories has been formulated in Refs.~\cite{JCB02,Song1,SDH-BFYM}. In this subsection, we present a unified construction of these HYM theories using the higher generalized forms.
	
	Ordinary Yang–Mills theory is formulated on a principal $G$-bundle, where $G$ is a compact simple Lie group. The special abelian case  $G=\text{U}(1)$  reduces  to Maxwell electrodynamics. On the base manifold $M$,  a connection is described by a $\mathfrak{g}$-valued 1-form $A\in \Lambda^1(M, \mathfrak{g})$.  Its curvature is the $\mathfrak{g}$-valued 2-form
	\begin{equation*}
		F=dA+\dfrac{1}{2}[A, A]\in \Lambda^2(M, \mathfrak{g}),
	\end{equation*}
	which  satisfies the Bianchi identity $dF+[A, F]=0$.
	The Yang–Mills action is then defined as the inner product of the curvature with itself,
	\begin{equation*}
		(\!( F, F)\!)=\int_M \langle F, *F\rangle_{\mathfrak{g}},
	\end{equation*}
	where the pairing is the one introduced in \eqref{inn0}.
	
	In parallel, 2YM theory Ref.~\cite{JCB02} is formulated on a principal 2-bundle whose gauge group is a Lie 2-group $(H, G; \bar{\alpha}, \bar{\triangleright})$. For a 2-connection $(A, B)$ with curvature $(\Omega_1, \Omega_2)$ as in \eqref{2-curvature}, the action can be written as the inner product of the generalized curvature $\mathcal{F}=\Omega_1+\Omega_2\xi$,
	\begin{align*}
		(\!(\mathcal{F}, \mathcal{F} )
		\!)&=(\!(\Omega_1+\Omega_2\xi, \Omega_1+\Omega_2\xi )
		\!)\\
		&=\int_M\big( \langle \Omega_1, *\Omega_1\rangle_{\mathfrak{g}} + \langle \Omega_2, *\Omega_2\rangle_{\mathfrak{h}} \big),
	\end{align*}
	which follows directly from the definition \eqref{g-inn1}.
	
	Proceeding to the next level, 3YM theory \cite{Song1} is defined on a principal 3-bundle with gauge group given by a Lie 3-group $(L, H, G;\bar{\beta}, \bar{\alpha}, \bar{\triangleright}, \left\{\cdot,\cdot\right\})$. For a 3-connection $(A, B, C)$ whose curvatures $(\Omega_1, \Omega_2, \Omega_3)$ are defined in \eqref{3-cu}, the action can be expressed as the inner product of the  generalized form of type $N=2$,
	\begin{equation*}
		\tilde{\mathcal{F}}= \Omega_1+ \dfrac{1}{2}\Omega_2\xi^1+\dfrac{1}{2} \Omega_2\xi^2+ \Omega_3\xi^1\xi^2,
	\end{equation*}
	namely
	\begin{align*}
		(\!(\tilde{\mathcal{F}}, \tilde{\mathcal{F}} )
		\!)=&(\!(\Omega_1+ \dfrac{1}{2}\Omega_2\xi^1+\dfrac{1}{2} \Omega_2\xi^2+ \Omega_3\xi^1\xi^2, \Omega_1+ \dfrac{1}{2}\Omega_2\xi^1\nonumber\\
		&+\dfrac{1}{2} \Omega_2\xi^2+ \Omega_3\xi^1\xi^2)
		\!)\\
		=&\int_M\big( \langle \Omega_1, *\Omega_1\rangle_{\mathfrak{g}} + \langle \Omega_2, *\Omega_2\rangle_{\mathfrak{h}} +\langle \Omega_3, *\Omega_3\rangle_{\mathfrak{l}}  \big),
	\end{align*}
	which follows directly from the definition of the inner product in \eqref{g-inn2}.
	
	The constructions above show that the generalized inner product   $(\!( -, - )\!)$ offers a unified way to formulate HYM actions. By applying this inner product to the corresponding generalized curvature forms, we recover the standard 2YM action from the type $N=1$ curvature $\mathcal{F}$, and the 3YM action from a suitably normalized type $N=2$ curvature $\tilde{\mathcal{F}}$.  
	This directly generalizes the ordinary Yang–Mills action, which is the type $N=0$. Consequently, the formalism developed here provides a  candidate for constructing Yang–Mills-like actions for higher gauge theories, where the action would be defined analogously as the generalized inner product of the appropriate higher curvature.

	\section{Conclusion and outlook}\label{sec-6}
	In this article, we have developed a unified formulation of higher gauge theory within the framework of GDC. In particular, the higher gauge structures associated with Lie $2$- and Lie $3$-groups are encoded by generalized forms of types $N=1$ and $N=2$, respectively. Our construction completes and extends the results of Ref.~\cite{SDH4} by incorporating gauge transformations via higher group-valued generalized $0$-forms, and by providing a uniform derivation of the action functionals for both HCS and HYM theories.
	
	The framework establishes a precise correspondence, summarized in Fig.~\ref{figure 1}, between the type index $N$ and the categorical order of the gauge theory. Ordinary gauge theory corresponds to type $N=0$. For $N=1$, $2$-gauge theory is encoded by a $2$-connection as a generalized $1$-form in $\mathbf{H}_{(1)}$, together with the $2$-gauge transformation as a generalized $0$-form in $\mathbf{G}_{(1)}$. Likewise, type $N=2$ describes $3$-gauge theory, with the corresponding objects living in $\mathbf{H}_{(2)}$ and $\mathbf{G}_{(2)}$. This hierarchy organizes both ordinary and higher theories within a single coherent framework, thereby facilitating the transfer of well-established techniques from ordinary gauge theory to higher gauge theories. More broadly, the formalism developed here indicates that GDC provides a modular and recursively extensible scheme for defining fields, symmetries, and action principles in higher gauge theory.
	\begin{figure*}[htp!] 
		\begin{center}
			\begin{tikzpicture}	
				\node[align=center] (G) at (0,0){\scriptsize $\textbf{G}_{(0)}$};
				\node[align=center] (G') at (-1,0){\scriptsize $N=0$:};
				\node[align=center] (H) at (2,0){\scriptsize $\textbf{H}_{(0)}$};
				
				\node[align=center] (G1) at (1,-1){\scriptsize $\textbf{G}_{(1)}$};
				\node[align=center] (G1') at (0,-1){\scriptsize $N=1$:};
				
				\node[align=center] (H1) at (3,-1){\scriptsize $\textbf{H}_{(1)}$};
				\node[align=center] (ar) at (2.7,0){\scriptsize $\Longrightarrow$};
				\node[align=center] (ar) at (6.8,0){\scriptsize  \textbf{gauge theory: connection, gauge transformation}};
				\node[align=center] (ar) at (3.8,-1){\scriptsize $\Longrightarrow$};
				\node[align=center] (ar) at (8.3,-1){\scriptsize \textbf{2-gauge theory: 2-connection, 2-gauge transformation}};
				\node[align=center] (ar) at (4.7,-2){\scriptsize $\Longrightarrow$};
				\node[align=center] (ar) at (9.2,-2){\scriptsize \textbf{3-gauge theory: 3-connection, 3-gauge transformation}};
				
				\node[align=center] (b0) at (1,0.1){\scriptsize $\blacktriangleright$};
				\node[align=center] (b1) at (2,-0.9){\scriptsize $\blacktriangleright$};
				\node[align=center] (b2) at (3,-1.9){\scriptsize $\blacktriangleright$};
				
				\node[align=center] (G2) at (2,-2){\scriptsize $\textbf{G}_{(2)}$};
				\node[align=center] (G2') at (1,-2){\scriptsize $N=2$:};
				\node[align=center] (H2) at (4,-2){\scriptsize $\textbf{H}_{(2)}$};
				\node[align=center] (dot) at (1,0){ };
				\node[align=center] (dot1) at (2,-1){ };
				\node[align=center] (dot2) at (3,-2){ };
				\node[align=center] (dot3) at (3,-2.2){ $\vdots$};
				\node[align=center] (dot3) at (3,-2.7){ $\vdots$};
				
				\node[align=center] (left01) at (-2,-0.6){};
				\node[align=center] (left12) at (-2,-1.3){};
				\node[align=center] (right01) at (5.5,-0.6){};
				\node[align=center] (right12) at (5.5,-1.3){};
				
				
				\draw[-] (G) -- (H) node[midway, above, sloped] {};
				\draw[-] (dot) -- (G1) node[midway, above, sloped] {};
				\draw[-] (G1) -- (H1) node[midway, above, sloped] {};
				\draw[-] (dot1) -- (G2) node[midway, above, sloped] {};
				\draw[-] (G2) -- (H2) node[midway, above, sloped] {};
			\end{tikzpicture}
		\end{center}
		\caption{Correspondence between generalized forms and (higher) gauge theories}\label{figure 1}
	\end{figure*}
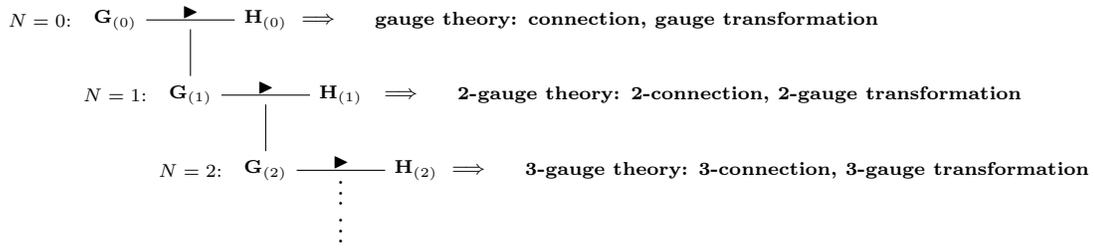

	Several natural directions for future research arise from this work. A first step is to extend the GDC encoding to higher types ($N>2$), and to establish general recursive formulae for higher connections, curvatures, and gauge transformations, together with the gauge invariance of the associated HCS and HYM-type actions. Further directions include supersymmetric extensions of the formalism and systematic applications to theories with higher-form gauge fields and symmetries, such as effective descriptions of topological phases and models relevant to quantum gravity.

	\section{Acknowledgments}
		We thank the Beijing International Center for Mathematical Sciences at Peking University for providing a supportive research environment.

	%

	\appendix
	\section{Higher groups and algebras}
	This section collects the fundamental definitions and relations that underpin the terminology and notation used in the main text. We  recall essential aspects of higher groups with a focus on strict Lie 2-groups and Lie 3-groups. These are described by crossed modules and 2-crossed modules of Lie groups, respectively. 
	Further details on higher groups can be found in Refs.~\cite{JFM,TRMV,HC,AMTPRB,YHMNRY,Martins:2010ry,TRMV2}. 
	\subsection{Lie 2-groups and Lie 2-algebras}\label{A.1}
	In a (strict) 2-gauge theory, the underlying algebraic structure is a Lie 2-group, which admits an equivalent description as a Lie crossed module. Accordingly, its infinitesimal counterpart, a Lie 2-algebra, is equivalently described by a differential crossed module.
	
	\textbf{Lie crossed modules.}
	A  Lie crossed module is a quadruple $(H, G; \bar{\alpha}, \bar{\triangleright})$ consisting of two Lie groups $H$ and $G$, a Lie group homomorphism $\bar{\alpha}: H\longrightarrow G$, and a smooth action $\bar{\triangleright}: G \times H \longrightarrow H$ of $G$ on $H$ by Lie group automorphisms. 
	These data satisfy the following compatibility conditions for all $g\in G$ and $h, h'\in H$,
	\begin{equation}\label{equi}
		\bar{\alpha} \left(g \bar{\triangleright} h\right) = g \bar{\alpha} \left(h\right) g^{-1},\qquad	\bar{\alpha} \left(h\right) \bar{\triangleright} h' = h h' h^{-1}.
	\end{equation}
	
	\textbf{Differential crossed modules.}
	A  differential crossed module  is a quadruple $(\mathfrak{h}, \mathfrak{g}; \alpha, \triangleright)$ consisting of two Lie algebras $\mathfrak{h}$ and $\mathfrak{g}$,  and  Lie algebra homomorphisms  $\alpha: \mathfrak{h}\longrightarrow \mathfrak{g}$ and $\triangleright: \mathfrak{g}\longrightarrow \operatorname{Der}(\mathfrak{h})$.  It is required that for all $X, X_1, X_2\in \mathfrak{g}$ and $Y, Y', Y_1, Y_2\in \mathfrak{h}$ the following hold:
	\begin{enumerate}
		\item The $\mathfrak{g}$-equivariance of $\alpha$ and the Peiffer identity,
		\begin{align}\label{YyY'}
			\alpha(X\triangleright Y)=[X, \alpha(Y)],\qquad
			\alpha(Y)\triangleright Y'=\left[Y,Y'\right].
		\end{align}
		As a consequence of \eqref{YyY'}, the kernel $\operatorname{ker} \alpha \subset \mathfrak{h}$ is an abelian ideal.
		\item The induced map $\triangleright: \mathfrak{g} \times \mathfrak{h} \longrightarrow \mathfrak{h}$ is bilinear and satisfies the derivation properties,
		\begin{align}
			X \triangleright \left[Y_1 ,  Y_2\right]  &= \left[ X \triangleright Y_1 , Y_2\right] +\left[  Y_1 , X \triangleright Y_2\right],\label{XY_1Y_2}\\
			\left[X_1, X_2\right] \triangleright Y &=X_1 \triangleright (X_2 \triangleright Y) - X_2 \triangleright(X_1 \triangleright Y). \label{X_1X_2Y}
		\end{align}
	\end{enumerate}
	
	The differential of a Lie crossed module $(H, G; \bar{\alpha}, \bar{\triangleright})$ yields a differential crossed module $(\mathfrak{h}, \mathfrak{g}; \alpha, \triangleright)$  \cite{Martins:2010ry}.  Here, $\mathfrak{g}$ and $\mathfrak{h}$ are the Lie algebras of $G$ and $H$, respectively, and the Lie algebra homomorphisms are obtained by differentiation,
	\begin{align*}
		\alpha=d\bar{\alpha}|_{1_H},\ \ \triangleright=d\bar{\triangleright}|_{1_G},
	\end{align*}
	where $\alpha$ is the differential of the homomorphism $\bar{\alpha}$ at the identity $1_H\in H$, $\triangleright$ is the differential of the adjoint action $\bar{\triangleright}: G \longrightarrow \operatorname{Aut}(H)$ at $1_G\in G$. The linear map $\triangleright: \mathfrak{g}\times \mathfrak{h}\longrightarrow \mathfrak{h}$ defined in this way automatically satisfies the derivation properties \eqref{XY_1Y_2} and \eqref{X_1X_2Y}.  Furthermore, differentiating the two compatibility conditions \eqref{equi} of the Lie crossed module gives precisely the $\mathfrak{g}$-equivariance and Peiffer identities \eqref{YyY'} of the differential crossed module.

	\textbf{Mixed relations.}
	Given a Lie crossed module $(H, G; \bar{\alpha}, \bar{\triangleright})$ and its associated differential crossed module $(\mathfrak{h}, \mathfrak{g}; \alpha, \triangleright)$,  the group $G$ acts on the Lie algebras $\mathfrak{g}$ and $\mathfrak{h}$ in a compatible way. 
	On $\mathfrak{g}$ the action is the adjoint action
	$Ad: G \times \mathfrak{g} \longrightarrow \mathfrak{g}$, defined by
	$Ad_g X= g X g^{-1}$.
	On $\mathfrak{h}$ the action is induced from the group action $\bar{\triangleright}$ and will be denoted by the same symbol
	$\triangleright: G \times \mathfrak{h} \longrightarrow \mathfrak{h}$. These actions satisfy the following identities for all
	$g$, $g_1$, $g_2 \in G$, $X \in \mathfrak{g}$, $Y \in \mathfrak{h}$ and $h \in H $,
	\begin{align}
		(g_1 g_2)\triangleright Y& = g_1 \triangleright (g_2 \triangleright Y),\label{A.6}\\
		g \triangleright (X \triangleright Y)&= (Ad_g X)\triangleright (g \triangleright Y) \label{A.7},\\
		\alpha(g\triangleright Y)&= Ad_g \alpha(Y),\label{gY2}\\
		\alpha (h) \triangleright Y &= h Y h^{-1}. \label{hy}
	\end{align}
	The compatibility relations \eqref{A.6}--\eqref{hy} are essential in the construction of 2-gauge field theories.
	
	\textbf{$G$-invariant pairing}. For a Lie crossed module $(H, G; \bar{\alpha}, \bar{\triangleright})$, there exists a non-degenerate pairing \cite{Z-2021,ASBV},
	$$\langle-, -\rangle_{\mathfrak{g}, \mathfrak{h}}: \mathfrak{g}\times \mathfrak{h}\longrightarrow \mathbb{R}$$ 
	on the underlying differential crossed module $(\mathfrak{h}, \mathfrak{g}; \alpha, \triangleright)$.
	This pairing is $G$-invariant, 
	\begin{align}\label{gin}
		\langle g X g^{-1}, g \triangleright Y\rangle_{\mathfrak{g},\mathfrak{h}} = \langle X, Y\rangle_{\mathfrak{g},\mathfrak{h}}, 
	\end{align}
	for all $ g \in G, X \in \mathfrak{g}, Y \in \mathfrak{h}$,
	and satisfies the symmetry condition
	\begin{equation}\label{symp}
		\langle \alpha(Y_1), Y_2\rangle_{\mathfrak{g}, \mathfrak{h}}=\langle \alpha(Y_2), Y_1\rangle_{\mathfrak{g}, \mathfrak{h}},
	\end{equation}
	for all $Y_1,Y_2\in\mathfrak{h}$.
	
	Differentiating the $G$-invariance \eqref{gin} yields the corresponding infinitesimal $\mathfrak{g}$-invariance,
	\begin{equation}\label{XXY}
		\langle[X_1, X_2], Y\rangle_{\mathfrak{g}, \mathfrak{h}}=-\langle X_2, X_1\triangleright Y\rangle_{\mathfrak{g}, \mathfrak{h}}
	\end{equation}
	for all $ X_1, X_2\in \mathfrak{g}$ and $Y \in \mathfrak{h}$.
	
	\textbf{$G$-invariant bilinear forms.}
	A symmetric, non-dege- \\nerate and $G$-invariant structure on $(\mathfrak{h}, \mathfrak{g}; \alpha, \triangleright)$  can also be given by a pair of separate bilinear forms. It consists of non-degenerate symmetric forms $\langle-, -\rangle_{\mathfrak{g}}$ on $\mathfrak{g}$ and $\langle-, -\rangle_{\mathfrak{h}}$ on $\mathfrak{h}$ satisfying:
	\begin{enumerate}
		\item $\langle-, -\rangle_{\mathfrak{g}}$ is $G$-invariant, 
		\begin{equation*}
			\langle \mathrm{Ad}_g X, \mathrm{Ad}_g X'\rangle_{\mathfrak{g}}=\langle X, X'\rangle_{\mathfrak{g}}, \ \ \ \forall g\in G, \ X, X' \in \mathfrak{g}.
		\end{equation*}
		\item $\langle-, -\rangle_{\mathfrak{h}}$ is $G$-invariant, 
		\begin{equation*}
			\langle g\triangleright Y, g\triangleright Y'\rangle_{\mathfrak{h}}=\langle Y, Y'\rangle_{\mathfrak{h}}, \ \ \ \forall g\in G, \ Y, Y' \in \mathfrak{h}.
		\end{equation*}
	\end{enumerate}
	Further details on such invariant structures can be found in Ref.~\cite{Martins:2010ry}.

	\subsection{Lie 3-groups and Lie 3-algebras}\label{A.2}
	The algebraic framework underlying a (strict) 3-gauge theory is a Lie 3-group, which is equivalently described by a Lie 2-crossed module. Correspondingly, the infinitesimal structure, a Lie 3-algebra, is equivalently represented by a differential 2-crossed module.
	
	\textbf{Lie 2-crossed modules.}
	A Lie 2-crossed module is a tuple $(L, H, G;\bar{\beta}, \bar{\alpha}, \bar{\triangleright}, \left\{\cdot,\cdot\right\})$ consisting of a complex of Lie groups
	$$L \stackrel{\bar{\beta}}{\longrightarrow}H \stackrel{\bar{\alpha}}{\longrightarrow}G$$
	together with a smooth left action $\triangleright$ of $G$ on $L$, $H$ and itself (by conjugation) by automorphisms, 
	\begin{align*}
		g \bar{\triangleright} (e_1 e_2)=(g \bar{\triangleright} e_1)(g \bar{\triangleright} e_2), \ \ \ (g_1 g_2)\bar{\triangleright} e = g_1 \bar{\triangleright} (g_2 \bar{\triangleright} e),
	\end{align*}
	for all $g, g_1, g_2\in G, e, e_1, e_2\in H $ or $L$, and a $G$-equivariant smooth map $ \left\{\cdot,\cdot \right\} :H \times H \longrightarrow L $, called the Peiffer lifting, satisfying
	\begin{equation*}
		g \bar{\triangleright} \left\{ h_1, h_2 \right\} = \left\{g \bar{\triangleright} h_1, g \bar{\triangleright} h_2\right\}
	\end{equation*}
	for all $g \in G, h_1,h_2\in H$.
	These data are required to fulfill the following axioms for all $h_i\in H$ and $l_i\in L$:
	\begin{enumerate}
		\setlength{\itemsep}{4pt}
		\item The maps $\bar{\beta}$ and $\bar{\alpha}$ are $G$-equivariant and satisfy $\bar{\alpha} \circ \bar{\beta}=1_G$ (the constant map to the identity of $G$).
		
		\setlength{\itemsep}{4pt}
		\item $\bar{\beta} \left\{h_1,h_2\right\} =\lbrack\lbrack h_1 , h_2 \rbrack\rbrack $, where $\lbrack\lbrack h_1 , h_2 \rbrack\rbrack=h_1h_2h^{-1}_1 (\bar{\alpha}(h_1)\bar{\triangleright} h^{-1}_2)$.
		
		\setlength{\itemsep}{4pt}
		\item $\left[l_1, l_2\right]= \left\{ \bar{\beta} (l_1), \bar{\beta} (l_2)\right\}$  with $\left[l_1,l_2\right]=l_1 l_2 l_1 ^{-1} l_2 ^{-1}$.
		
		\setlength{\itemsep}{4pt}
		\item $\left\{ h_1 h_2, h_3 \right\}= \left\{h_1, h_2 h_3 h_2 ^{-1}\right\} \bar{\alpha} (h_1) \bar{\triangleright} \left\{h_2,h_3\right\}$.
		
		\setlength{\itemsep}{4pt}
		\item $ \left\{h_1,h_2 h_3\right\}  =  \left\{ h_1 , h_2  \right\} \left\{ h_1 , h_3  \right\} \left\{ \lbrack\lbrack h_1, h_3 \rbrack\rbrack  ^{-1}, \bar{\alpha} (h_1) \bar{\triangleright} h_2 \right\}$.
		
		\setlength{\itemsep}{4pt}
		\item $\left\{ \bar{\beta} (l) , h\right\} \left\{h, \bar{\beta} (l)\right\} = l (\bar{\alpha} (h)\bar{\triangleright} l^{-1})$.
	\end{enumerate}
	From these axioms one can define a left action $\bar{\triangleright}' $ of $H$ on $L$ by
	\begin{equation*}
		h \bar{\triangleright}' l = l \left\{ \bar{\beta} (l)^{-1}, h \right\}, \qquad \forall h\in H, l \in L,
	\end{equation*}
	which, together with the homomorphism 
	$ \bar{\beta} : L \longrightarrow H $, endows
	$(L, H; \bar{\beta}, \bar{\triangleright}')$ with the structure of a crossed module. In particular, one has $h \bar{\triangleright}' 1_L = \left\{ 1_H, h\right\} =\left\{h, 1_H\right\}=1_L$  for each $h\in H$.

	\textbf{Differential 2-crossed modules.}
	A differential 2-crossed module is a tuple $(\mathfrak{l},\mathfrak{h}, \mathfrak{g}; \beta, \alpha,\triangleright, \left\{ \cdot,\cdot \right\})$ consisting of a complex of Lie algebras
	\[\mathfrak{l}\stackrel{\beta}{\longrightarrow} \mathfrak{h} \stackrel{\alpha}{\longrightarrow} \mathfrak{g},\]
	together with a left action $\triangleright $ of $\mathfrak{g}$ on $\mathfrak{l}, \mathfrak{h}$ and itself (by the adjoint representation) by derivations, and a $\mathfrak{g}$-equivariant bilinear map $\left\{\cdot,\cdot \right\}:\mathfrak{h} \times \mathfrak{h} \longrightarrow \mathfrak{l}$, the (infinitesimal) Peiffer lifting, satisfying
	\begin{align}\label{12}
		X\triangleright \left\{ Y_1,Y_2\right\} = \left\{X\triangleright Y_1,Y_2\right\} + \left\{Y_1, X\triangleright Y_2\right\}
	\end{align}
	for all $X \in \mathfrak{g}, Y_1,Y_2 \in \mathfrak{h}$.
	
	The following axioms hold for all $X \in \mathfrak{g}, Y_i\in \mathfrak{h}$ and $Z_i \in \mathfrak{l}$:
	\begin{enumerate}
		\setlength{\itemsep}{4pt}
		\item 	$\mathfrak{l}\stackrel{\beta}{\longrightarrow} \mathfrak{h} \stackrel{\alpha}{\longrightarrow} \mathfrak{g}$ is a complex of $\mathfrak{g}$-modules and satisfies  $\alpha \circ \beta =0 $.
		
		\setlength{\itemsep}{4pt}
		\item $\beta \left\{Y_1,Y_2\right\} =\lbrack\lbrack Y_1 , Y_2\rbrack\rbrack $, where $\lbrack\lbrack Y_1 , Y_2\rbrack\rbrack=\left[ Y_1, Y_2\right] - \alpha (Y_1) \triangleright Y_2$.
		
		\setlength{\itemsep}{4pt}
		\item $ \left[Z_1, Z_2\right]= \left\{ \beta (Z_1), \beta (Z_2)\right\}$.
		
		\setlength{\itemsep}{4pt}
		\item $\left\{ \left[Y_1 ,Y_2\right], Y_3 \right\}= \alpha (Y_1) \triangleright \left\{Y_2,Y_3\right\}+\left\{Y_1,\left[Y_2,Y_3\right]\right\}-\alpha (Y_2)\triangleright \left\{Y_1,Y_3\right\}-\left\{Y_2,\left[Y_1,Y_3\right]\right\}$. Equivalently,
		\begin{align*}
			\left\{\left[Y_1,Y_2\right],Y_3\right\}=& \left\{\alpha (Y_1)\triangleright Y_2, Y_3\right\} - \left\{ \alpha (Y_2)\triangleright Y_1, Y_3\right\}
			- \left\{Y_1, \beta \left\{ Y_2, Y_3\right\}\right\} + \left\{Y_2, \beta\left\{Y_1,Y_3 \right\}\right\}.
		\end{align*}

		\setlength{\itemsep}{4pt}
		\item $ \left\{Y_1,\left[Y_2,Y_3\right]\right\}= \left\{ \beta\left\{Y_1,Y_2 \right\},Y_3  \right\}-\left\{ \beta\left\{Y_1,Y_3 \right\},Y_2 \right\}$.
		
		\setlength{\itemsep}{4pt}
		\item $\left\{ \beta (Z) , Y \right\} +\left\{Y, \beta (Z)\right\} =- \alpha (Y)\triangleright Z$.
	\end{enumerate}
	Analogously to the group case, one defines a left action of $\mathfrak{h}$ on $\mathfrak{l}$ by
	\begin{equation}\label{YZZ}
		Y\triangleright'Z= -\left\{\beta(Z), Y\right\}, \qquad \forall Y\in \mathfrak{h}, Z \in \mathfrak{l}, 
	\end{equation}
	so that $(\mathfrak{l}, \mathfrak{h}; \beta, \triangleright')$
	becomes a differential crossed module. In addition, if the original action satisfies
	\begin{align}\label{A.17}
		\alpha (Y)\triangleright Z= Y\triangleright'Z,
	\end{align}
	the differential 2-crossed module is called \emph{fine}.
	
	In the special case where $\mathfrak{h}$ is Abelian and $\alpha$ is trivial, the above axioms reduce to the following conditions:
	\begin{subequations}\label{B.7}
		\begin{align}
			\beta \left\{Y_1,Y_2\right\} &=0,\\
			\left[Z_1, Z_2\right]&= \left\{ \beta (Z_1), \beta (Z_2)\right\},\\
			\left\{ \beta (Z) , Y \right\} &=-\left\{Y, \beta (Z)\right\}.
		\end{align}
	\end{subequations}
	The present work deals exclusively with this simplified setting. For a comprehensive treatment of 3-groups we refer the reader to Ref.~\cite{doi:10.1063/1.4870640}.

	\textbf{Mixed relations.}
	Let $(L, H, G;\bar{\beta}, \bar{\alpha}, \bar{\triangleright}, \left\{\cdot,\cdot\right\})$ be a Lie 2-crossed module and $(\mathfrak{l},\mathfrak{h},\mathfrak{g};\beta,
	\alpha,\triangleright,\left\{\cdot,\cdot\right\})$ the associated differential 2-crossed module. 
	Besides the mixed relations already listed for the underlying crossed module $(H, G; \bar{\alpha}, \bar{\triangleright})$  (see Eqs.~\eqref{A.6}--\eqref{hy}), the action of $G$ on $\mathfrak{l}$ satisfies the following identities for all $g, g_1, g_2 \in G$, $X \in \mathfrak{g}$, $Y_i\in \mathfrak{h}$ and $Z \in \mathfrak{l}$:
	\begin{align}
		\beta(g \triangleright Z)&=g \triangleright \beta(Z),\label{gZ}\\
		(g_1 g_2)\triangleright Z &= g_1 \triangleright (g_2 \triangleright Z),\\
		g \triangleright (X \triangleright Z)&= (Ad_g X)\triangleright (g \triangleright Z) ,\label{A.24}\\
		g \triangleright\{Y_1, Y_2\}&=\{g \triangleright Y_1, g \triangleright Y_2\}. \label{gyy}
	\end{align}

	\textbf{$G$-invariant pairing}. 
	For a 2-crossed module as above, a $G$-invariant pairing  consists of two bilinear forms:
	\begin{enumerate}
		\item An antisymmetric, non-degenerate bilinear form $\langle - , - \rangle_\mathfrak{h}:  \mathfrak{h} \times \mathfrak{h} \longrightarrow \mathbb{R}$ satisfying
		\begin{align}
			\langle [Y,Y_1], Y_2 \rangle_\mathfrak{h}&=-\langle Y_1,[Y,Y_2]\rangle_\mathfrak{h},\\
			\langle X\triangleright Y, Y_1\rangle_\mathfrak{h}&=-\langle Y, X\triangleright Y_1 \rangle_\mathfrak{h}.\label{YX}
		\end{align}
		\item A non-singular bilinear form $\langle - , - \rangle_{\mathfrak{g},\mathfrak{l}}: \mathfrak{g} \times \mathfrak{l} \longrightarrow \mathbb{R}$ satisfying
		\begin{align}
			\langle[X_1, X_2], Z\rangle_{\mathfrak{g},\mathfrak{l}} &= - \langle X_2, X_1 \triangleright  Z\rangle_{\mathfrak{g},\mathfrak{l}} \label{XZ},\\
			\langle \alpha(Y), Z \rangle_{\mathfrak{g},\mathfrak{l}} &= - \langle \beta(Z), Y\rangle_{\mathfrak{h}},\label{YZ}\\
			\langle X, \{ Y_1, Y_2\} \rangle_{\mathfrak{g}, \mathfrak{l}}&=\frac{1}{2} \langle Y_2, X \triangleright Y_1 \rangle_{\mathfrak{h}},\label{XYY}
		\end{align}
		for $X, X_1, X_2 \in \mathfrak{g}$, $Y, Y_1, Y_2 \in \mathfrak{h}$ and $Z \in \mathfrak{l}$.
	\end{enumerate}
	
	The non-singularity of $\langle -, - \rangle_{\mathfrak{g}, \mathfrak{l}}$ implies that the 2-crossed module is \emph{balanced}, i.e., $\dim \mathfrak{l}= \dim \mathfrak{g}$. Both forms are required to be $G$-invariant,
	\begin{align*}
		\langle g\triangleright Y,  g\triangleright Y' \rangle_\mathfrak{h}&= \langle Y,Y'\rangle_\mathfrak{h},\\
		\langle g X g^{-1}, g \triangleright Z\rangle_{\mathfrak{g},\mathfrak{l}} &= \langle X, Z\rangle_{\mathfrak{g},\mathfrak{l}},
	\end{align*}
	for any $g \in G, X \in \mathfrak{g}, Y, Y'\in \mathfrak{h}$ and $ Z \in \mathfrak{l}$.
	
	\textbf{$G$-invariant bilinear forms.}
	A symmetric non-degenerate and $G$-invariant structure on a differential 2-crossed module can also be specified by a triple of symmetric bilinear forms: $\langle-, -\rangle_{\mathfrak{g}}$ on $\mathfrak{g}$, $\langle-, -\rangle_{\mathfrak{h}}$ on $\mathfrak{h}$, and $\langle-, -\rangle_{\mathfrak{l}}$ on $\mathfrak{l}$. Each form is required to be non-degenerate and invariant under the respective $G$-action:
	\begin{enumerate}
		\item $\langle-, -\rangle_{\mathfrak{g}}$ is $G$-invariant, 
		\begin{equation*}
			\langle \mathrm{Ad}_g X, \mathrm{Ad}_g X'\rangle_{\mathfrak{g}}=\langle X, X'\rangle_{\mathfrak{g}}, \ \ \ \forall g\in G, \ X, X' \in \mathfrak{g}.
		\end{equation*}
		\item $\langle-, -\rangle_{\mathfrak{h}}$ is $G$-invariant, 
		\begin{equation*}
			\langle g\triangleright Y, g\triangleright Y'\rangle_{\mathfrak{h}}=\langle Y, Y'\rangle_{\mathfrak{h}}, \ \ \ \forall g\in G, \ Y, Y' \in \mathfrak{h}.
		\end{equation*}
		\item $\langle-, -\rangle_{\mathfrak{l}}$ is $G$-invariant, 
		\begin{equation*}
			\langle g\triangleright Z, g\triangleright Z'\rangle_{\mathfrak{l}}=\langle Z, Z'\rangle_{\mathfrak{l}}, \ \ \ \forall g\in G, \ Z, Z' \in \mathfrak{l}.
		\end{equation*}
	\end{enumerate}
	Further details on such invariant structures can be found in Refs.~\cite{TRMV,TRMV1}.
	
	\section{Differential forms valued in (higher) algebras}\label{AVDDF}
	We then introduce ordinary differential forms that take values in the  strict Lie 2- and Lie 3-algebras (i.e., in differential crossed modules and 2-crossed modules), following the framework introduced in our earlier work \cite{SDH-BFYM}.
	
	Throughout this paper we adopt the Einstein summation convention: repeated indices are summed over their entire range.   Let $\mathfrak{g}$ be a Lie algebra of dimension $m$ with basis $\{X_a\}_{a=1}^m$. 
	A $\mathfrak{g}$-valued $p$-form $A$ on a manifold $M$ is an expression of the form  
	\begin{align*}
		A=A^a\otimes X_a
		=\dfrac{1}{p!}A^a_{\mu_1\cdots \mu_p}dx^{\mu_1}\wedge \cdots \wedge dx^{\mu_p}\otimes X_a,
	\end{align*}
	where  each component $A^a$ is an ordinary $p$-form and the coefficients $A^a_{\mu_1, \dots, \mu_p}$ are smooth functions. 
	
	Standard operations on differential forms extend to the Lie algebra-valued case in a natural way. The exterior derivative
	\[d: \Lambda^p(M, \mathfrak{g})\longrightarrow\Lambda^{p+1}(M, \mathfrak{g})\]
	is calculated by 
	\begin{equation*}
		dA=dA^a\otimes X_a.
	\end{equation*}
	For two forms $A = A^a \otimes X_a \in \Lambda^p(M, \mathfrak{g})$ and $\bar{A} = \bar{A}^b \otimes X_b \in \Lambda^q(M, \mathfrak{g})$,
	the Lie bracket combines the exterior product of forms with the Lie bracket of $\mathfrak{g}$,
	\begin{align*}
		[A, \bar{A}]=A^a\wedge \bar{A}^b\otimes [X_a, X_b].
	\end{align*}
	This definition implies the graded symmetry
	\begin{align*}
		[A, \bar{A}]=(-1)^{pq+1}[\bar{A}, A],
	\end{align*}
	and the exterior derivative satisfies the graded Leibniz rule
	\begin{equation*}
		d[A, \bar{A}]=[dA, \bar{A}]+(-1)^p[A, d\bar{A}].
	\end{equation*}
	The same conventions apply to differential forms valued in other Lie algebras, such as $\mathfrak{h}$ or $\mathfrak{l}$.

	In particular, for the higher algebras that appear in this work, the structure maps are extended as follows.
	If $(\mathfrak{h}, \mathfrak{g}; \alpha, \triangleright)$  is a differential crossed module and $B = B^b \otimes Y_b \in \Lambda^t(M, \mathfrak{h})$, we set
	\begin{equation*}
		\alpha(B)=B^b\otimes \alpha (Y_b) \in \Lambda^t(M, \mathfrak{g}).
	\end{equation*}
	For a differential 2-crossed module
	$(\mathfrak{l}, \mathfrak{h}, \mathfrak{g}; \beta, \alpha, \triangleright, \{\cdot,\cdot \})$ and a form $C = C^c \otimes Z_c \in \Lambda^v(M, \mathfrak{l})$, 
	\begin{equation*}
		\beta(C)=C^c\otimes \beta (Z_c) \in \Lambda^v(M, \mathfrak{h}).
	\end{equation*}
	
	Similarly, the action $\triangleright$ of $\mathfrak{g}$ and the Peiffer lifting are defined by
	\begin{align*}
		A\triangleright B&=A^a\wedge B^b\otimes X_a \triangleright Y_b,\\ A\triangleright C&=A^a\wedge C^c\otimes X_a \triangleright Z_c,\\
		\{B_1, B_2\}&=B^{b_1}\wedge B^{b_2}\otimes \{Y_{b_1}, Y_{b_2}\},
	\end{align*}
	where $A = A^a \otimes X_a \in \Lambda^p(M, \mathfrak{g})$ and $B_i= B^{b_i} \otimes Y_{b_i }\in \Lambda^{t_i}(M, \mathfrak{h})$, $i=1, 2$, and $C = C^c \otimes Z_c \in \Lambda^v(M, \mathfrak{l})$.
	
	Finally, the $G$-invariant pairings introduced in the preceding subsections are extended to differential forms in the same component-wise fashion. For the pairing  $\langle - , - \rangle_{\mathfrak{g},\mathfrak{h}}$ and forms $A\in \Lambda^p(M, \mathfrak{g})$, $B\in \Lambda^q(M, \mathfrak{h})$,
	\begin{align*}
		\langle A, B \rangle_{\mathfrak{g}, \mathfrak{h}}=A^a\wedge B^b\langle X_a, Y_b\rangle_{\mathfrak{g}, \mathfrak{h}}.
	\end{align*}
	For $\langle - , - \rangle_{\mathfrak{g},\mathfrak{l}}$ and forms $A\in \Lambda^p(M, \mathfrak{g})$, $C\in \Lambda^q(M, \mathfrak{l})$,
	\begin{align*}
		\langle A, C \rangle_{\mathfrak{g}, \mathfrak{l}}=A^a\wedge C^c\langle X_a, Z_c\rangle_{\mathfrak{g}, \mathfrak{l}}.
	\end{align*}
	For a symmetric invariant form $\langle - , -\rangle_{\mathfrak{w}}$ on a Lie algebra $\mathfrak{w}$  (where $\mathfrak{w}= \mathfrak{g},  \mathfrak{h}$ or $\mathfrak{l}$)
	with basis $\{T_a\}_{a=1}^m$,
	\begin{align*}
		\langle W_1, W_2 \rangle_{\mathfrak{w}}=W_1^a\wedge W_2^b\langle T_a, T_b\rangle_{\mathfrak{w}}, 
	\end{align*}
	where $W_1, W_2\in \Lambda(M, \mathfrak{w})$.
	Further properties of these extended operations, as required in the main text, are derived straightforwardly from the definitions above.
	%

		%
	%


\begin{thebibliography}{99}
			
			\bibitem{Baez.2010}J.C. Baez and J. Huerta, Gen. Rel. Grav. \textbf{43}, 2335--2392  (2011).
			
			\bibitem{FH}F. Girelli and H. Pfeiffer, J. Math. Phys. \textbf{45},  3949--3971 (2004).
			
			\bibitem{Baez2005HigherGT} 
			J.C. Baez and U. Schreiber, arXiv.math/0511710 (2007).
			
			\bibitem{Borsten.2024}
			L. Borsten, M.J. Farahani, B. Jurco, H. Kim, J. Narozny, D. Rist, C. Saemann and M. Wolf, \textit{Encyclopedia of Mathematical Physics} (Second Edition) \textbf{4}, 159--185  (2025).
			
			\bibitem{JCBADL}
			J.C. Baez and A.D. Lauda, arXiv:math/0307200 (2004).
			
			\bibitem{JCBASC}J.C. Baez and A.S. Crans, arXiv:math/0307263 (2004).
			
			\bibitem{JBUS}
			J.C. Baez and U. Schreiber, arXiv:hep-th/0412325 (2004).
			
			
			\bibitem{doi:10.1063/1.4870640}  
			W. Wang, J. Math. Phys. \textbf{55}, 043506 (2014).
			
			\bibitem{JFM} J.F. Martins and R. Picken, Differ. Geom. Appl. \textbf{29}(2), 179--206 (2011).
			
			\bibitem{TRMV}T. Radenkovi\'c and M. Vojinovi\'c, JHEP \textbf{2019}, 222 (2019).
			
			
			\bibitem{PLB}P. Aschieri, L. Cantini and B. Jurco, Commun. Math. Phys. \textbf{254}, 367--400 (2005).
			%
			\bibitem{Wockel}
			C. Wockel, Forum Math. \textbf{23}(3),  565--610 (2011).
			%
			\bibitem{Stevenson}
			T. Nikolaus and K.Waldorf, arXiv:1103.4815 (2013).
			%
			\bibitem{BJ}
			B. Jurco, Int. J. Geom. Meth. Mod. Phys. \textbf{8}(1), 49--78 (2011).
			%
			\bibitem{Nikolaus11}
			T. Nikolaus, U. Schreiber and D. Stevenson,  J. Homot. Relat. Struct. \textbf{10}, 749--801 (2015).
			%
			\bibitem{Bunk}
			S. Bunk, Math. Ann. \textbf{387}, 689--743 (2023).
			%
			
			\bibitem{Leinster}
			T. Leinster, TAC \textbf{10}, 1--70 (2002).
			%
			\bibitem{Simpson}
			C. Simpson, \textit{Homotopy Theory of Higher Categories} (Cambridge University Press, Cambridge, 2011).
			%
			
			
			\bibitem{Pfeiffer}
			H. Pfeiffer, Ann. Phys. \textbf{308}(2),  447--477 (2003).
			
			\bibitem{Henneaux}
			M. Henneaux and C. Teitelboim, Found. Phys. \textbf{16}, 593--617 (1986).
			
			
			\bibitem{Gastel}
			A. Gastel, Commun. Math. Phys. \textbf{376}, 1053--1071 (2020).
			
			
			\bibitem{JCB02}
			J.C. Baez, arXiv:hep-th/020613 (2002).
			%
			\bibitem{Song1}
			D.H. Song, K. Lou, K.Wu, J. Yang and F.H. Zhang, J. Geom. Phys. \textbf{178}, 104537 (2022).
			%
			\bibitem{SDH-BFYM}
			D.H. Song, K. Lou, K.Wu and J. Yang, Eur. Phys. J. C \textbf{82}, 1034 (2022).
			%
			\bibitem{Z-2021} R. Zucchini, JHEP \textbf{06}, 025 (2021).
			%
			%
			\bibitem{Zucchini-2014-1}
			E. Soncini and R. Zucchini, JHEP \textbf{10}, 079 (2014).
			%
			\bibitem{Zucchini-2016}
			R. Zucchini, J. Math. Phys. \textbf{57}, 052301 (2016).
			%
			\bibitem{Zucchini-2013}
			R. Zucchini, JHEP \textbf{03}, 014 (2013).
			%
			%
			\bibitem{SDH4} D.H. Song, M.Y. Wu, K.Wu and J. Yang, JHEP \textbf{07}, 207 (2023).
			%
			
			\bibitem{2007R3} D.C. Robinson, J. Phys. A \textbf{40}, 8903--8922 (2007).
			
			\bibitem{2001NR} P. Nurowski and D.C. Robinson, Class. Quant. Grav. \textbf{18}, 81--86 (2001).
			
			\bibitem{2002NR}P. Nurowski and D.C. Robinson,  Class. Quant. Grav.  \textbf{19}, 2425--2436 (2002).
			%
			\bibitem{1998GAJ}
			G.A.J. Sparling, \textit{The geometric universe: science, geometry, and the work of Roger Penrose} (Oxford University Press, Oxford, 1998). 
			%
			\bibitem{Cartan0}
			E. Cartan, Annales Sci. Ecole Norm. Sup. \textbf{16}, 239--332 (1899).
			%
			%
			%
			%
			\bibitem{Cartan2}
			E. Cartan, Annales Sci. Ecole Norm. Sup. \textbf{21}, 153--206 (1904).
			%
			%
			%
			\bibitem{Cartan5}
			E. Cartan, Ann. Soc. Polon. Math. \textbf{8}, 181--225 (1929).
			%
			%
			%
			\bibitem{1998ZPGAJS}Z. Perjes and G.A.J. Sparling, Twistor Newsletter \textbf{44},  12--19 (1998).
			%
			\bibitem{GAJ} G.A.J. Sparling, University of Pittsburgh preprint (1997).
			%
			\bibitem{2003R1}D.C. Robinson,  J. Math. Phys. \textbf{44},  2094--2110 (2003).
			%
			\bibitem{2003R2}D.C. Robinson, Int. J. Theor. Phys. \textbf{42},  2971--2981 (2003).
			%
			\bibitem{2002GLTZ}H.Y. Guo, Y. Ling, R.S. Tung and Y.Z. Zhang, Phys. Rev. D \textbf{66},  064017 (2002). 
			%
			\bibitem{2004LTG}Y. Ling, R.S. Tung and H.Y. Guo, Phys. Rev. D \textbf{70},  044045 (2004). 
			%
			%
			\bibitem{2009R}D.C. Robinson, Class. Quant. Grav. \textbf{26},  075019 (2009).
			%
			\bibitem{2013R}D.C. Robinson, arXiv:1312.0846 (2013).
			%
			%
			%
			\bibitem{2008SCALANS}S. Chatterjee, A. Lahiri and A.N. Sengupta, Int. J. Geom. Meth. Mod. Phys. \textbf{5}(4),  573--586 (2008). 
			%
			
			\bibitem{NM}
			M. Nakahara, \textit{Geometry, Topology and Physics} (Chemical Rubber Company Press, Boca Raton 2003).
			
			%
			\bibitem{HKCS}
			H. Kim and C. Saemann, J. Phys. A \textbf{53}, 445206 (2020).
			
			
			\bibitem{K.W} K.Waldorf, arXiv:1704.08542 (2018).
			%
			%
			\bibitem{SZ} E. Soncini and R. Zucchini, J. Geom. Phys. \textbf{95},  28--73 (2015).
			%
			\bibitem{Zucchini-2020-1}
			R. Zucchini, J. Geom. Phys. \textbf{156}, 103826 (2020).
			%
			\bibitem{Zucchini-2020-2}
			R. Zucchini, J. Geom. Phys. \textbf{156}, 103825 (2020).
			%
			%
			
			
			%
			\bibitem{Jurco2019}B. Jur\v{c}o, T. Macrelli, L. Raspollini, C. S\"{a}mann and M. Wolf, Fortsch. Phys. \textbf{67},  8--9 (2019).
			%
			%
			\bibitem{ASBV}A. Schenkel and B. Vicedo, Commun. Math. Phys. \textbf{405},  293 (2024).
			%
			%
			%
			\bibitem{FGAO}F. Girelli, A. Osumanu and W. Wieland, Phys. Rev. D \textbf{105},  044003 (2022).
			%
			\bibitem{HC}H. Chen and F. Girelli, Phys. Rev. D \textbf{106}, 105017 (2022).
			%
			\bibitem{DSF}
			D.S. Freed, Adv. Math. \textbf{113},  237--303 (1995).
			%
			%
			%
			\bibitem{TP}
			T.R. Govindarajan and P. Ramadevi, \textit{Geometry and topology of low dimensional systems: Chern–Simons theory with applications} (Springer, Berlin, 2024).
			%
			\bibitem{HC-2024}H. Chen and J. Liniado, Phys. Rev. D \textbf{110},  086017 (2024).
			%
			\bibitem{JCB-1996}J.C. Baez, Lett. Math. Phys. \textbf{38}, 129--143 (1996).
			%
			\bibitem{MGFGCG} M. Geiller, F. Girelli, C. Goeller and P. Tsimiklis, JHEP \textbf{05},  154 (2023).  
			%
			\bibitem{FGPT}F. Girelli and P. Tsimiklis, Phys. Rev. D \textbf{106},  046003 (2022).
			%
			%
			%
			%
			\bibitem{AMTPRB}
			A. Mutlu, T. Porter and R. Brown, Theor. Appl. Categ. \textbf{4}, 174  (1998).
			%
			\bibitem{YHMNRY}Y. Hidaka, M. Nitta and R. Yokokura, JHEP \textbf{01},  173 (2021).
			%
			\bibitem{Martins:2010ry} J.F. Martins and A. Mikovic, Adv. Theor. Math. Phys. \textbf{15},  1059--1084 (2011).
			%
			\bibitem{TRMV2}
			T. Radenkovic and M. Vojinovic, Class. Quant. Grav. \textbf{39},  135009 (2022).
			%
			%
			\bibitem{TRMV1}T. Radenkovic and M. Vojinovic, Ann. U. Craiova Phys. \textbf{30},  74--84 (2020).
			
		\end{thebibliography}
\end{document}